\documentclass[letterpaper]{article} % DO NOT CHANGE THIS
\usepackage[preprint]{aaai2027}  % DO NOT CHANGE THIS
\usepackage[hyphens]{url}  % DO NOT CHANGE THIS
\usepackage{graphicx} % DO NOT CHANGE THIS
\usepackage{natbib}  % DO NOT CHANGE THIS AND DO NOT ADD ANY OPTIONS TO IT
\usepackage{caption} % DO NOT CHANGE THIS AND DO NOT ADD ANY OPTIONS TO IT
\usepackage{algorithm}
\usepackage{algorithmic}

\usepackage{newfloat}
\usepackage{listings}
\DeclareCaptionStyle{ruled}{labelfont=normalfont,labelsep=colon,strut=off} % DO NOT CHANGE THIS
\floatstyle{ruled}
\newfloat{listing}{tb}{lst}{}
\floatname{listing}{Listing}

\usepackage{booktabs}
\usepackage{comment}
\usepackage{amsmath,amsthm,amssymb}
\newcommand{\red}[1]{\textcolor{red}{#1}}
\newtheorem{assumption}{Assumption}
\newtheorem{proposition}{Proposition}
\newtheorem*{remark}{Remark}
\usepackage{tabularx}
\usepackage{subcaption}
\title{Learning Long-Term Educational Investment Policies under Residential Sorting}
\author{
    Hpnglei Guo\textsuperscript{\rm 1},
    Shuo Chen\textsuperscript{\rm 1},
    Yuhan Zhao\corresponding \textsuperscript{\rm 1}
}
\affiliations{
    \textsuperscript{\rm 1}Zhe Jiang University\\
    \textsuperscript{\rm 2}BIGAI
}

\title{Learning Long-Term Educational Investment Policies under Residential Sorting}
\author {
    Honglei Guo \textsuperscript{\rm 1,\rm 2},
    Shuo Chen \textsuperscript{\rm 2},
    Mingjie Bi\textsuperscript{\rm 2},
    Zeyang Sun\textsuperscript{\rm 2},
    Xiaoxi Wang\textsuperscript{\rm 2},
    Yuhan Zhao\textsuperscript{\rm 2}%\corresponding
}
\affiliations {
    \textsuperscript{\rm 1}  College of Artificial Intelligence,
  Zhejiang University\\
    \textsuperscript{\rm 2}  State Key Laboratory of General Artificial Intelligence, BIGAI\\
   zhaoyuhan@bigai.ai%, secondAuthor@affilation2.com, thirdAuthor@affiliation1.com
}

\begin{document}

\maketitle

\begin{abstract}
Allocating public-school investment effectively and fairly is difficult when school access depends on residence. School improvements can raise nearby housing demand and prices, reshape enrollment, and potentially limit access for lower-income households. These effects evolve as residential sorting changes school composition, quality, and future investment needs. Existing approaches often study school funding, household choice, and housing markets separately, while static models can miss their interconnected, long-term effects. We address this gap with a dynamic multi-agent framework that links government investment, household sorting, housing prices, population turnover, enrollment, and evolving school quality. A government planner uses reinforcement learning (RL) to identify multiyear allocation policies that account for household responses while balancing aggregate educational access and equity. In simulations, our RL-based policy attains the highest access level (0.4780) and second-lowest access Gini coefficient (0.0164) among representative baselines, demonstrating a favorable effectiveness–equity balance. The results also indicate reduced socioeconomic stratification in educational access. By making education–housing feedback explicit, our framework supports long-term analysis of how school investment shapes educational opportunity over time.
\end{abstract}

% Uncomment the following to link to your code, datasets, an extended version or similar.
% You must keep this block between (not within) the abstract and the main body of the paper.
% Make sure that you do not de-anonymize yourself with these links.
% \begin{links}
%     \link{Code}{https://aaai.org/example/code}
%     \link{Datasets}{https://aaai.org/example/datasets}
%     \link{Extended version}{https://aaai.org/example/extended-version}
% \end{links}

\section{Introduction} \label{sec:intro}

\begin{figure*}
    \centering
    \includegraphics[width=0.9\linewidth]{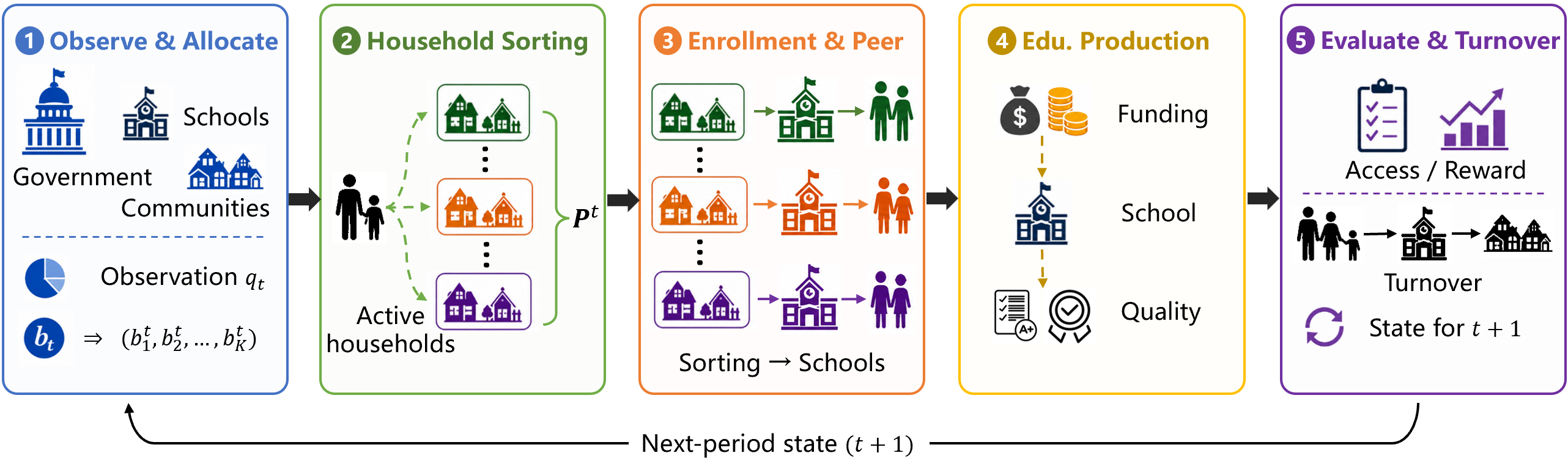}
    \caption{Overview of the multi-agent resource allocation framework. Government resource allocation affects school access indirectly through household sorting, housing-market adjustment, enrollment, and school-quality dynamics. These interactions jointly determine the next-period state, and the government uses reinforcement learning to learn a sequential allocation policy.
    %Illustration of the multi-agent framework for educational resource allocation. Different entities take five interactive blocks each periods to demonstrate the impact of educational resources. The government leverages RL to seek a cross-year optimal allocation strategies.
    }
    \label{fig:intro}
\end{figure*}

Access to high-quality education is central for long-term social opportunity. However, in many public-school-dominant education systems, access to educational quality is highly impacted by residential location and school-district assignment. When high-quality schools are concentrated in particular communities, households with sufficient resources may compete for housing in these areas, while less advantaged households may face limited access to the same educational opportunities \cite{chetty2016effects,biasi2023school,owens2018income}. 
This makes educational resource allocation an important but difficult social planning problem. Governments must allocate limited resources across schools or districts while balancing multiple objectives, such as %improving average educational quality, 
reducing inequality in access and maintaining efficient resource use. 
In this work, we operationalize these educational resources as government investment at the school level, which represents policy inputs that can improve school quality, such as teacher development, facility improvement, and other forms of school-level support.

The difficulty of educational investment is that its effect does not stop at the school receiving the resource. It propagates through a dynamic social loop. Government investment changes school quality; school quality affects household residential and school choices; household choices change housing prices, student enrollment, and the demographic composition of schools; these changes then reshape inequality, school demand, and future allocation needs. 
%
%This complicated social loops contains four features that need to be addressed properly. First, educational resource is coupled with the housing market when residential location determines access to schools. Second, households are heterogeneous in income and preferences, which can lead to various school choices. Third, the effects of investment are not immediate and persist over time because improvements in education and student outcomes emerge gradually. Lastly, households, schools, housing markets, and governments interact together, so the outcome of an allocation policy depends on the responses of multiple actors rather than on the government's decision alone.
Planning within this loop is challenging for several reasons. When residential location determines school access, educational investment becomes coupled with the housing market, so improvements in school quality may also change local housing demand and prices. Household responses are also heterogeneous because families differ in income and preferences, leading the same investment to affect access differently across groups. Moreover, the effects of investment unfold gradually and can persist across multiple periods. These mechanisms interact simultaneously, meaning that the outcome of an allocation decision depends not only on the government's action but also on the endogenous responses of households, schools, and housing markets.

Prior work has examined important parts of this system. For example, \citet{barseghyan2019peer} studied household behavior model with a focus on peer preferences in public school choice; \citet{abdulkadirouglu2003school} approached the school choice problem via mechanism design; \citet{black1999better} and \citet{caetano2019neighborhood} investigated the relationship between school quality and housing prices; while \citet{monarrez2023school} and \citet{owens2016income} focused on educational inequality and segregation.
These studies characterize important links in the social loop, but they typically examine them separately rather than as part of a unified dynamic resource-allocation problem. This leaves open how a government should allocate resources when school quality, residential sorting, housing prices, enrollment, and educational access co-evolve over time.

To address this gap, we propose a dynamic multi-agent framework for educational investment allocation. As illustrated in Figure~\ref{fig:intro}, the framework models households and government as interacting agents, while representing schools, the housing market, and demographic change as key environment components. 
Households differ in socioeconomic characteristics and preferences, choose residential locations under housing and educational constraints, and generate evolving enrollment patterns. Schools receive government investment and update their educational quality over time. The government observes the evolving system and allocates limited investment across schools while balancing multiple objectives.
Building on this framework, we formulate government allocation as a sequential decision-making problem and use reinforcement learning (RL) to learn allocation policies.
We use simulations to compare different allocation strategies and examine their effects on educational access, resource efficiency, residential and school segregation, and inequality over time. We further conduct sensitivity analyses to study how key factors, including household preferences and school educational performance, shape policy outcomes.
Our contributions are as follows:
\begin{itemize}
    \item We develop a multi-agent computational framework that captures the dynamic education--housing feedback loop, linking government investment, school quality, household sorting, 
    %housing prices, 
    enrollment, and educational access.
    \item We characterize household sorting through a convex program and develop efficient algorithms for computing the resulting equilibrium.
    \item We develop an RL-based allocation approach that enables a government planner to learn multi-objective investment policies under delayed effects and endogenous household responses.
    % \item \red{social analysis?}
    \item We conduct comprehensive simulations based on the framework and conduct social impact analysis to show that our RL-based policy achieves a favorable balance between efficiency and equity by reducing disparities in school quality while preserving socioeconomic matching. %This promotes more equitable access to schools while increasing per-student investment.
\end{itemize}

\section{Related Work} \label{sec:related_work}
% educational resource allocation related. ABM. In this work, we consideration interactions and impacts. More comprehensive.

% urban planning with multi-agent and RL.

\paragraph{School Access and Household Sorting.} School-choice research studies assignment mechanisms and how preferences for school effectiveness and peer composition shape enrollment \cite{abdulkadirouglu2003school,abdulkadiroglu2020parents,barseghyan2019peer}. A related structural literature models residential sorting as an equilibrium outcome of heterogeneous household preferences and endogenous housing prices. \citet{epple1999estimating} estimate equilibrium sorting across local jurisdictions, while \citet{bayer2004equilibrium} integrate heterogeneous discrete residential choice with urban housing-market equilibrium. \citet{bayer2007unified} extend this framework to estimate preferences for schools and neighborhood attributes. Empirical evidence further shows that school quality is capitalized into housing prices and that school districts and attendance boundaries affect residential and school segregation \cite{black1999better,caetano2019neighborhood,turnbull2021meta,owens2017racial,owens2016income,monarrez2023school}. More integrated equilibrium models connect community composition, school spending, housing, peer effects, and migration \cite{fernandez1996income,nechyba1999school}. Recent spatial and agent-based models represent additional feedback between residential and school choices \cite{agostinelli2024spatial,dignum2022mechanisms,dignum2024data}. Our framework builds directly on the equilibrium-sorting literature but embeds household sorting within a dynamic system in which government investment changes school quality, housing demand, enrollment, and future allocation needs.

\paragraph{School Investment and Resource Planning.} Causal evidence from school-finance reforms shows that sustained spending can improve educational and long-term economic outcomes, often with larger benefits for disadvantaged students and effects that emerge gradually \cite{jackson2016effects,lafortune2018school,biasi2023school,jackson2024impacts}. Dynamic general-equilibrium analysis has also examined the long-run distributional and welfare consequences of school-finance reform \cite{fernandez1998public}. Computational work addresses different planning decisions. \citet{mayerle2022optimal} jointly match students, schools, teachers, classes, and classrooms, while attendance-boundary optimization balances integration, travel, and enrollment constraints \cite{gillani2023redrawing}. \citet{guan2025contextual} incorporate predicted school choices into stochastic rezoning optimization, and \citet{zhang2026deep} use deep RL for multistep school-district adjustment. These methods optimize finance reform, resource matching, or geographic assignment. Our planner instead allocates a recurring school-level investment budget and evaluates effectiveness and fairness after endogenous changes in quality, residential sorting, housing prices, and enrollment.

\paragraph{Closest Comparisons.} Three recent studies are especially close to ours. \citet{agostinelli2024spatial} jointly model housing, school access, and heterogeneous residential sorting, but evaluate school-choice expansion and housing vouchers rather than repeated school investment. \citet{guan2025contextual} use predicted household school choices to inform a segregation-oriented boundary redesign, whereas our household response operates through a housing-market sorting equilibrium after each investment decision. \citet{zhang2026deep} also use RL for multistep school planning, but their actions modify district boundaries and their objectives concern distance and utilization. Our RL planner allocates investment and evaluates the long-term distribution of accessed school quality. The main distinction is the combination of sequential investment, endogenous education--housing feedback, and an explicit effectiveness--equity objective.

\section{Agent Models and Environment Components} \label{sec:model}

The model is built around an institutional feature of public education: residential location determines school access. When government investment improves a school, it may also increase nearby housing demand and the price of entering that district. The same policy can therefore raise school quality while changing which households can access it. We capture this feedback through two decision-making agents, households and the government, and three environment components: communities and housing markets, population turnover, and schools.

\subsubsection{Setting and Notation.}
We study a region with $M$ residential communities and $K$ public schools over periods $t=0,\ldots,T$. A fixed district map $\mathcal D:\{1,\ldots,M\}\to\{1,\ldots,K\}$ identifies the school serving each community. Thus, moving to community $m$ provides access to school $\mathcal D(m)$.
Community $m$ has a fixed amenity value $Q_m$, housing capacity $C_m$, current vacancies $V_m^t$, and housing price $P_m^t$. These quantities distinguish school access from other neighborhood amenities and housing costs. School $k$ has quality $r_k^t$, and $\mathbf r^t=(r_1^t,\ldots,r_K^t)$ denotes the regional school quality profile. Housing capacity limits immediate access to desirable schools. Instead, households compete for units released through turnover, and prices mediate this competition.

\subsubsection{Household Model.}
A household's residential choice depends on its types/characteristics (e.g., income, parental education, child performance) and stage in the school life cycle. We describe household $i$ by a state tuple
\begin{equation}
    s_i^t = (\theta_i,c_i^t,m_i^t),
    \qquad
    \theta_i=(y_i,z_i,e_i)\in\Theta.
\end{equation}
Here, $\theta_i$ denotes the type vector, which includes income $y_i$, parental education $z_i$, and the child's persistent latent ability $e_i$. $c_i^t$ is the child's age, and $m_i^t$ is the current community. 
Household types follow a finite distribution $\{ p_\theta \}_{\theta \in \Theta}$ estimated from China Education Panel Survey data~\cite{ceps2015}.

Households consider school access only when a child enters school. Children below age six occupy housing but are not yet enrolled. At age six, household $i$ joins the entry cohort $\mathcal{I}^t = \{i: c_i^t = 6\}$ and may stay or move. Once enrolled, the household remains in its chosen community for $W$ school periods. Thus, only part of the population can respond to a new allocation in any period.

For a school-entry household, action $a_i^t$ selects a community, including the option to remain in $m_i^t$. The choice trades off school quality, housing cost, moving cost $D$, and neighborhood amenities, captured by the utility
\begin{equation}
\label{eq:utility}
\begin{split}
    &u_{im}^t := u_i^t(a_i^t = m) \\
    &= \underbrace{\alpha_i r_{\mathcal D(m)}^t + \eta Q_m -\left[ \beta_i P_m^t +D \right] \mathbf{1}(m\neq m_i^t)}_{\bar{u}_{im}^t} + \varepsilon_{im}^t.
\end{split}
\end{equation}
The housing price and moving cost apply only to movers because an incumbent retains its current unit. The shocks $\varepsilon_{im}^t \sim \mathrm{Gumbel}(0,\tau)$ models unobserved household-specific considerations and are independent across households, communities, and periods. 
Detailed specifications of parameters and utility design are provided in Appendix~\ref{app:model-spec.household_param}.

New households entering the region provide a second source of housing demand. Let $\mathcal J^t$ denote these preschool-age in-migrants. They do not yet make a school-entry decision, so their current utility depends on amenities and housing prices:
\begin{equation}
\label{eq:incoming-utility}
    v_{im}^t:=v_i^t(a_i^t=m)
    =\underbrace{\eta Q_m-\beta_iP_m^t}_{\bar v_{im}^t}
    +\varepsilon_{im}^{\mathrm{in},t}.
\end{equation}
Their shocks satisfy the same independence and distributional assumptions as those in~\eqref{eq:utility}. In-migrants compete for vacant units and therefore affect the prices faced by $\mathcal I^t$. They may reconsider their residential locations when their children later enter school.

\subsubsection{Population Turnover.}
At the start of period $t$, turnover from the preceding period has formed the resident population $\mathcal H^t$, released vacancies, and identified $\mathcal I^t$ and $\mathcal J^t$. Eligible non-school-age households may depart, and new preschool-age households are drawn from $p_\theta$. Each departure releases one unit in the household's community, so vacancies arise through turnover rather than being reset exogenously.

The population module maintains $\sum_{i\in\mathcal H^t}\mathbf 1(m_i^t=m)\leq C_m$ and scales in-migration so that expected new demand does not exceed aggregate vacancies. This makes residential composition slow-moving and prevents a policy from immediately reshuffling all students after quality changes. Appendix~\ref{app:model-spec.population} provides the transition details.

\subsubsection{School Model.}
Schools are modeled as education-production components. They take government allocation and household choices as input and generate education performance measured by school qualities. Let $\mathcal L_k^t$ be the children enrolled in school $k$ after period-$t$ sorting, and let $N_k^t=|\mathcal L_k^t|$. The mean ability of enrolled children is 
\begin{equation}
    S_k^t=\frac{1}{N_k^t}\sum_{i\in\mathcal L_k^t}e_i.
\label{eq:composition}
\end{equation}
Thus, peer composition changes endogenously as households select across communities. 
Given per-student investment $h_k^t$, school quality updates after sorting according to
\begin{equation}
\label{eq:quality}
    r_k^{t+1} = (1-\delta) r_k^t + \delta \left[ w_v g_v(h_k^t) + w_s S_k^t \right].
\end{equation}
The persistence parameter $\delta$ reflects gradual change in school conditions. The increasing, concave function $g_v$ captures diminishing returns to investment, while $w_v$ and $w_s$ weight resources and peer composition. We normalize $r_k^t\in(0,1]$ and $S_k^t, g_v(h_k^t)\in[0,1]$, and set $w_v+w_s=1$, so $r_k^t$ is a comparable quality index.

For a child living in community $m_i^t$, we denote current accessed quality as $\rho_i^t=r_{\mathcal D(m_i^t)}^t$. We also define the child learning-outcome index by
\begin{equation}
\label{eq:outcome}
    o_i^t=\rho_i^t e_i^\psi.
\end{equation}
This index combines accessed quality with persistent ability, which can be used to measure a child's received educational benefit. The exponent $\psi$ controls how persistent ability enters the learning-outcome index.
%We use it only as a secondary evaluation measure, not as an empirical grade or the government's policy objective.
%Appendix~\ref{app:model-details} provides more formal convention. \red{app}

\subsubsection{Government Model.}
The government cannot directly assign households to schools. Its policy instrument is an annual budget $B_t$ divided among the $K$ schools. Define the allocation simplex as $\Delta_K = \{ \mathbf{b} \in \mathbb{R }_+^K: \sum_{k=1}^K b_k = 1 \}$. The government chooses shares $\mathbf b^t\in\Delta_K$, so school $k$ receives total funding $g_k^t = b_k^t B_t$. After enrollment, per-student investment is $h_k^t = g_k^t/ N_k^t$. It is clear that $\sum_{k=1}^K g_k^t=B_t$. 
%
%Since the shares sum to one, increasing one school's share necessarily reduces the resources available to other schools. 
The central difficulty to determine the budget share is temporal. Current investment changes future school quality, future school quality changes residential demand, and residential choices changes who benefits from subsequent investment. We formalize this sequence and the government's objective in the next section.

\section{Multi-Agent Learning Framework} \label{sec:problem}

The components above interact once in every simulated period. Turnover at the end of period $t-1$ creates the start-of-period population and housing state at $t$. The feedback is illustrated in Figure~\ref{fig:intro}.
We describe this interaction as a dynamic leader--follower framework with lagged household response. The decision process can be divided into two levels. 
In the lower level, households compete for available residential spots given school quality $\mathbf{r}^t$, which leads to a sorting equilibrium.
In the upper level, the government anticipates the sorting equilibrium and determine the budget allocation $\mathbf{b}_t$, which changes the quality observed by households in later periods. 

\subsection{Lower-Level Household Sorting}
\subsubsection{Logit Demand.}
Only communities with vacancies can accept a mover. Let $\mathcal M^t=\{m:V_m^t>0\}$. A school-entry household chooses between staying and moving, with feasible set $\mathcal A_i^t=\{m_i^t\}\cup\mathcal M^t$; an in-migrant must choose from $\mathcal M^t$. Given $\mathbf P^t=(P_1^t,\ldots,P_M^t)$ and school quality $\mathbf r^t$, the Gumbel shocks in~\eqref{eq:utility}--\eqref{eq:incoming-utility} imply (See Appendix~\ref{app:equilibrium.logit})
\begin{align}
    \sigma_{im}^t(\mathbf P^t;\mathbf r^t)
    &=\frac{\exp(\bar u_{im}^t/\tau)}
    {\sum_{m'\in\mathcal A_i^t}\exp(\bar u_{im'}^t/\tau)},
\label{eq:logit} \\
    \mu_{im}^t(\mathbf P^t)
    &=\frac{\exp(\bar v_{im}^t/\tau)}
    {\sum_{m'\in\mathcal M^t}\exp(\bar v_{im'}^t/\tau)}.
\label{eq:bg}
\end{align}
These probabilities describe utility-maximizing choices under independent unobserved preferences. Summing them gives expected fresh demand for vacant units in community $m$:
\begin{equation}
\begin{split}
    D_m^t(\mathbf P^t;\mathbf r^t)
    =&\sum_{\substack{i\in\mathcal I^t\\m\neq m_i^t}}
    \sigma_{im}^t(\mathbf P^t;\mathbf r^t)
    +\sum_{i\in\mathcal J^t}\mu_{im}^t(\mathbf P^t).
\end{split}
\label{eq:demand}
\end{equation}
The first term is demand from school-entry movers; the second is demand from in-migrants. Households that stay do not consume vacancies.

\subsubsection{Market Clearing.}
We make the following assumption.
\begin{assumption}[Slack supply and price floor] \label{ass:slack}
In every period, expected fresh demand does not exceed aggregate vacancies,
$\sum_mD_m^t\leq\sum_mV_m^t$. Moreover, housing prices satisfy
$P_m^t\geq P_{\min}>0$.
\end{assumption}

The first condition allows the region to contain vacant units even when particular districts are contested. The second gives housing a minimum reservation value. Communities with excess capacity remain at the price floor, while contested communities get a price premium until expected demand matches vacancies. Thus, the clearing condition becomes a complementarity system: 
for each $m\in\mathcal M^t$, we have
\begin{equation}
\label{eq:clearing}
\begin{split}
    D_m^t(\mathbf P^t;\mathbf r^t) \leq V_m^t,
    \qquad P_m^t\geq P_{\min},\\
    (P_m^t-P_{\min}) [V_m^t-D_m^t(\mathbf P^t;\mathbf r^t)] = 0.
\end{split}
\end{equation}

A \emph{sorting equilibrium} is a triple $(\mathbf{P}^{t\star}, \boldsymbol{\sigma}^{t\star}, \boldsymbol{\mu}^{t\star})$ in which both household groups choose according to~\eqref{eq:logit}--\eqref{eq:bg} at prevailing prices and expected demand satisfies~\eqref{eq:clearing}. Prices summarize competition from other households; once prices are fixed, the logit probabilities determine choice behavior. 
This equilibrium is central to the social-impact question because it captures whether school improvements are capitalized into housing prices that limit access. 
We provide a convex-program characterization and Algorithm~\ref{alg:sorting} to efficiently compute the equilibrium in Appendix~\ref{app:equilibrium.convex}. After obtaining the equilibrium, a capacity-feasible realization procedure then maps these probabilities to discrete residential assignments (See Appendix~\ref{app:equilibrium.assignment}).

\subsection{Upper-Level Allocation}
\subsubsection{Government Objective.}
The government rewards both the aggregate level and the distribution of publicly supplied school quality. Let $\rho_i^{t,+}=r_{\mathcal D(m_i^t)}^{t+1}$ be the end-of-period quality accessible to a child enrolled after period-$t$ sorting. For every household $i \in \mathcal{L}^t:=\bigcup_k \mathcal{L}_k^t$, we define
\begin{equation}
\label{eq:swf}
    \mathcal{W}_\epsilon( \boldsymbol{\rho}^{t,+})=
    \left(\frac{1}{|\mathcal L^t|}
    \sum_{i\in \mathcal{L}^t}(\rho_i^{t,+})^{1-\epsilon}
    \right)^{1/(1-\epsilon)},
    \quad \epsilon \geq0.
\end{equation}
At $\epsilon=0$, the objective is average accessed quality, which represents aggregate effectiveness under the fixed annual budget. Larger $\epsilon$ gives more priority to children with lower access, and the limit as $\epsilon\to\infty$ is maximin access. Parameter $\epsilon$ therefore makes the effectiveness--equity judgment explicit. For the special case where $\epsilon=1$, we use the continuous geometric mean extension.

%The corresponding secondary outcome is $o_i^{t,+} = \rho_i^{t,+} e_i^\psi$.

Access is the primary welfare basis because it represents the school quality supplied through public policy.
Across the horizon, the government seeks an observation-based policy $\pi$ that maximizes
\begin{equation}
\label{eq:gov-objective}
    \mathbb E_\pi\left[
    \sum_{t=0}^{T-1}\gamma^t
    \left\{\mathcal W_\epsilon(\boldsymbol\rho^{t,+})
    -\lambda\|\mathbf b^t-\mathbf b^{t-1}\|_2^2\right\}
    \right],
\end{equation}
with $\mathbf b_{-1}$ properly initialized. The optimization is subject to population turnover, household sorting, school-quality transitions, and the annual budget. Discount factor $\gamma$ controls the weight on future cohorts, and $\lambda\geq0$ discourages abrupt changes in allocation shares.

\begin{remark}
Access is the primary welfare basis because it represents the school quality supplied through public policy. The government does not optimize $\mathbf o^t$ directly because this index reflects persistent child ability, which is more appropriate to be used as a measure of educational outcomes. 
\end{remark}

\subsubsection{Learning Allocation Policy.}
The government’s problem is sequential and we represent it as a partially observable Markov decision process (POMDP). 
The full state $x^t\in\mathcal X$ contains the household-level population and the variables required for the transition to be Markov. The government receives only a summarized observation
\begin{equation}
\label{eq:observation}
    q^t = \mathcal{O}(x^t)=
    \left[\mathbf r^t,
    \left(\frac{\operatorname{occ}_m^t}{C_m}\right)_{m=1}^M,
    \overline{\mathbf y}^{t},
    \overline{\mathbf e}^{t},
    \boldsymbol\chi^t\right],
\end{equation}
where $\operatorname{occ}_m^t$ is community occupancy, $\overline{\mathbf y}^{t}$ and $\overline{\mathbf e}^{t}$ are community-level mean income and ability, and $\boldsymbol\chi^t$ summarizes the entry group.
The action is $\mathbf b^t\in\Delta_K$. We parameterize the observation-based policy as
\begin{equation}
\label{eq:dirichlet-policy}
    \mathbf b^t \sim \pi_\phi(\cdot \mid q^t)
    =\mathrm{Dirichlet}\!\left(\boldsymbol\kappa_\phi(q_t)\right),
\end{equation}
so every sampled allocation satisfies the annual budget constraint. After each action, the environment computes sorting, enrollment, the school transition, welfare, and turnover. The one-period reward is the corresponding term in~\eqref{eq:gov-objective}. We use PPO~\cite{schulman2017proximal} to learn a policy over full trajectories. Household sorting is solved within each environment step rather than learned by the government. %Algorithm~\ref{alg:system} gives the training loop.
Appendix~\ref{app:gov-decision} provides POMDP characterization and the training algorithm.

%\red{maybe remove}
%The framework does not assume that investment necessarily reduces or increases inequality. Instead, it makes the relevant pathways observable: which schools improve, who can access them, how housing prices mediate access, and whether gains persist across cohorts.

\section{Simulations and Experiments} \label{sec:exps}

\paragraph{Setup.} We initialize the environment with 4 schools and 12 communities, and the corresponding school districts are shown in Figure~\ref{fig:env-setting}. Each community has 100 residential positions. Household profiles, including income, educational level, child ability, are sampled from the CEPS data. Their utility coefficients are computed based on sampled profiles. The data preprocessing and parameter calibration are provided in Appendix~\ref{app:experiment.data}.

\begin{figure}[t]
\centering
    \includegraphics[width=0.9\columnwidth]{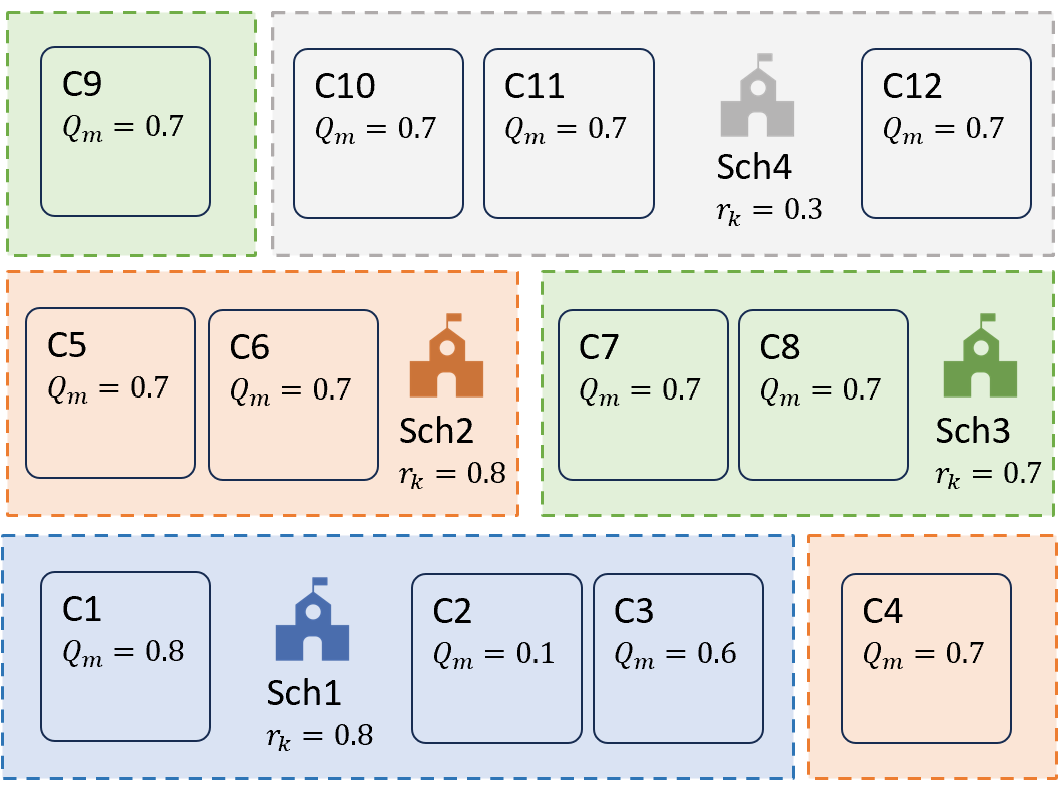}
    \caption{Initial environment settings. 12 communities (C1–C12) with amenity $Q_m$ are allocated to 4 school (Sch1-Sch4) districts, distinguished by colors. Households are spatially distributed among communities and schools provide heterogeneous educational quality $r_k$. This setting serves as the basis for evaluating different allocation policies.}
    \label{fig:env-setting}
\end{figure}

\paragraph{Metrics.} We propose six complementary metrics summarized in Table~\ref{tab:evaluation-metrics} to evaluate efficiency and equity. 
Mean Access (MA) and Human Capital (HC) measure overall allocation efficiency. 
Gini Access (GA) and Income Access Gap (IA) measure inequality in educational opportunities, while Income Dissimilarity (ID) and Income--Quality Correlation (IQ) capture socioeconomic segregation and the association between household income and school quality, respectively. 
%Table~\ref{tab:evaluation-metrics} summarizes their interpretations. 
Detailed definitions and implementations are provided in Appendix~\ref{app:experiment.metrics}.

\begin{table}[H]
    \centering
    \begin{tabular}{@{}lp{0.8\columnwidth}@{}}
        \toprule
        \textbf{Metric} & \textbf{Interpretation} \\
        \midrule
        MA & Average school quality accessed by households \\
        HC & Aggregate educational outcomes \\
        GA & Inequality in access to school quality \\
        IA & Access gap between income groups \\
        ID & Residential segregation across income groups \\
        IQ & Association between household income and school quality \\
        \bottomrule
    \end{tabular}
    \caption{Summary of evaluation metrics.}
    \label{tab:evaluation-metrics}
\end{table}

To evaluate the effectiveness, robustness, and scalability of our framework, we conduct four sets of experiments: (i) comparisons with representative baseline resource-allocation policies, (ii) module-level ablation studies, (iii) parameter sensitivity analyses, and (iv) scalability evaluations across increasing problem sizes.

%\begin{comment}
\subsection{Comparison with Allocation Baselines}
\label{sec:exps.comparison}

We compare the PPO allocation policy with three transparent policy rules. The \emph{equal-split} policy assigns the same total funding to every school; the \emph{enrollment-proportional} policy allocates funding according to current enrollment, thereby equalizing per-student investment; and the \emph{compensatory} policy assigns larger budget shares to lower-quality schools. Appendix~\ref{app:experiment.comparison} provides the formal
definitions and evaluation protocol.

Table~\ref{tab:evaluation-transpose} reports long-term average performance experimented by 100 seeds. The PPO policy provides the \emph{strongest balance} on the planner's primary objectives, achieving the highest Mean Access (MA; 0.478) and the second-lowest Gini Access (GA; 0.0164). The MA values are close across all policies (0.477--0.478), but the proposed policy reaches the highest access level while maintaining substantially lower inequality than equal-split and compensatory funding. Its Human Capital is also within 0.19\% of the best result (106.39 versus 106.60), indicating that the improved access balance does not require a substantial loss in overall educational outcomes.

\begin{table}[t]
\centering
\scriptsize
\setlength{\tabcolsep}{3.2pt}
\begin{tabular}{@{}lcccc@{}}
\toprule
Metric $\backslash$ Policy & Equal & Enroll Prop. & Compensatory & PPO \\
\midrule
MA {\tiny $(\times 10^{-2})$} $\uparrow$
  & 47.77(0.27) & 47.77(0.32) & 47.78(0.36) & \textbf{47.80(0.27)} \\
HC $\uparrow$
  & 106.59(1.58) & 106.28(1.58) & \textbf{106.60(1.67)} & 106.40(1.61) \\
\midrule
GA {\tiny $(\times 10^{-2})$} $\downarrow$
  & 2.29(0.21) & \textbf{0.59(0.07)} & 1.99(0.21) & 1.64(0.18) \\
IA {\tiny $(\times 10^{-3})$} $|\cdot|\downarrow$
  & -5.41(1.43) & \textbf{0.39(0.51)} & -4.51(1.16) & 6.00(1.34) \\
\midrule
ID {\tiny $(\times 10^{-1})$} $\downarrow$
  & 1.41(0.14) & 1.44(0.16) & \textbf{1.39(0.14)} & 1.47(0.15) \\
IQ {\tiny $(\times 10^{-1})$} $|\cdot|\downarrow$
  & -2.95(0.60) & \textbf{0.56(0.92)} & -2.90(0.59) & 3.99(0.66) \\
\bottomrule
\end{tabular}
\caption{Long-term average performance across allocation policies. 
%Bold indicates the best value. 
For IA and IQ, values closer to zero indicate greater equality between income groups.}
\label{tab:evaluation-transpose}
\end{table}

The income-based measures provide an important qualification. The IA remains close to zero under every policy (-0.0054--0.006), and ID also lies within a narrow range (0.139--0.147). The PPO policy therefore does not create a large income access gap, although its positive IQ (0.399) indicates that higher-income communities remain more strongly associated with higher-quality schools. Figure~\ref{fig:income-quality} illustrates this relationship.

These findings distinguish equality in the overall distribution of school quality from equality between income groups. The proposed policy improves the former, which is directly represented in the government objective, while its advantages are more limited for the latter. This is socially important because it shows that equalizing school-quality access alone may not eliminate the residential mechanisms through which income differences persist. The framework makes this distinction visible and can support future evaluations in which income-based fairness or housing access is incorporated explicitly into policy design.

\begin{figure}[t]
    \centering
    \includegraphics[width=1\columnwidth]
        {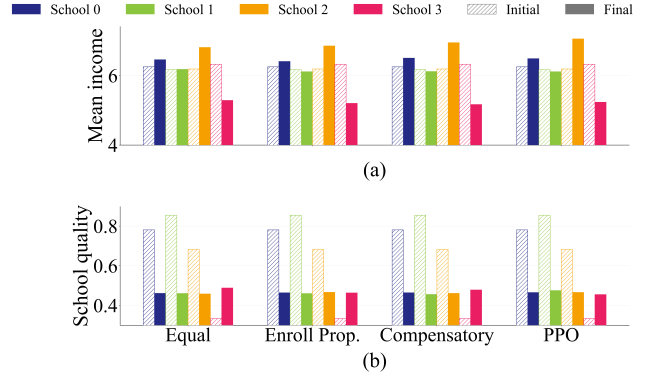}
    \caption{The policy impact on (a) school-level mean household income and (b) school quality before and after policy simulations.}
    \label{fig:income-quality}
\end{figure}

%Figure~\ref{fig:comparison-error-bars} summarizes repeated simulations under different random initializations. 
Besides, The limited variation within each policy in Table~\ref{tab:evaluation-transpose} shows that the main findings are not driven by a particular simulated population. In particular, the proposed policy consistently maintains high average access and relatively low access inequality. Because the numerical differences in MA are small, these results support a stable balance between access and equality rather than a large advantage on every individual metric.

% \begin{figure}[t]
%     \centering
%     \includegraphics[width=\columnwidth]
%         {Figures/evaluation_error_bar_chart.png}
%     \caption{Mean performance across random seeds. Error bars denote
%     \emph{[standard deviation/standard error/95\% confidence interval]}.}
%     \label{fig:comparison-error-bars}
% \end{figure}

\subsection{Ablation Study}
\label{sec:exps.ablation}

We examine the contributions of two social mechanisms: student-composition feedback in school-quality evolution and endogenous housing-price formation.
The student-composition feed mechanism affects peer composition on school quality evolution in \eqref{eq:quality} and removing it freezes $S^t_k$ in every simulation periods. 
%The housing-price formation corresponds to the endogenous housing price. 
Removing the housing-price formation mechanism fixes housing prices and eliminates the feedback between household choices and residential sorting.
Table~\ref{tab:ablation-results-transpose} compares the model without either mechanism, with each mechanism separately, and with both mechanisms. 

\begin{table}[t]
\centering
\small
\setlength{\tabcolsep}{3.2pt}
\begin{tabular}{@{}lcccc@{}}
\toprule
Metric $\backslash$ Mode & Neither & \shortstack{Composition\\only}
       & \shortstack{Prices\\only} & Full model \\
\midrule
MA {\tiny $(\times 10^{-2})$} $\uparrow$
  & 46.38 & 46.34 & 46.39 & \textbf{46.41} \\
HC $\uparrow$
  & 111.77 & 111.50 & 111.67 & \textbf{111.99} \\
\midrule
GA {\tiny $(\times 10^{-2})$} $\downarrow$
  & 2.08 & 2.08 & 1.26 & \textbf{1.22} \\
IA {\tiny $(\times 10^{-3})$} $|\cdot|\downarrow$
  & 0.10 & \textbf{0.00} & 0.10 & 0.30 \\
\midrule
ID {\tiny $(\times 10^{-1})$} $\downarrow$
  & 0.67 & \textbf{0.64} & 1.38 & 1.39 \\
IQ {\tiny $(\times 10^{-1})$} $|\cdot|\downarrow$
  & \textbf{0.05} & \textbf{0.05} & -0.06 & 0.13 \\
\bottomrule
\end{tabular}
\caption{Performance under different ablation settings. Bold indicates the
best value. For IA and IQ, values closer to zero indicate greater equality
between income groups.}
\label{tab:ablation-results-transpose}
\end{table}

The full model achieves the highest MA (0.4641) and HC (111.99) and the lowest GA (0.0122). The MA values remain within a narrow range (0.4634--0.4641), so the main benefit is not a large increase in one outcome but consistent performance across access, educational outcomes, and access inequality. When housing prices are endogenous, adding student-composition feedback produces modest improvements in MA ($+0.2\times 10^{-3}$), HC ($+0.314$), and GA ($-0.4\times 10^{-3}$). This suggests that composition feedback is most useful when it operates within the complete education--housing interaction.

The housing-price mechanism produces the more socially consequential change. When student-composition feedback is present, allowing prices to respond more than doubles ID (from 0.064 to 0.139), even as GA falls from 0.0208 to 0.0122. Thus, endogenous housing prices can make school-quality access more evenly distributed overall while simultaneously strengthening income sorting across locations. The apparently lower segregation obtained when prices are fixed should therefore be interpreted cautiously: it results partly from removing the channel through which demand for better schools affects residential costs.

Figure~\ref{fig:results} supports this interpretation dynamically. Housing prices respond to household demand only when endogenous price formation is included, while student-composition feedback changes how school quality evolves after households sort. Including both mechanisms enables the framework to represent an important social tension: school investment may improve overall access while its benefits remain mediated by residential markets.

\begin{figure}[t]
\centering
\includegraphics[width=\columnwidth]
    {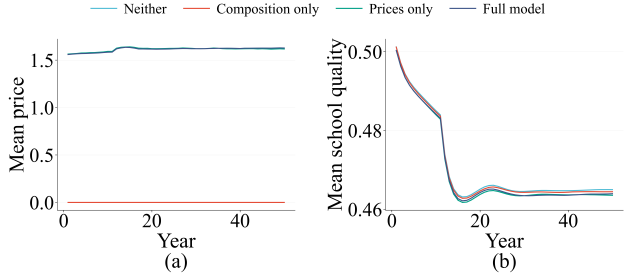}
    \caption{Evolution of (a) mean housing price and (b) school quality under the ablation settings.}
\label{fig:results}
\end{figure}

\subsection{Sensitivity and Scale Analysis}
\label{sec:exps.sensitivity-scale}

We first vary one parameter at a time while keeping the remaining parameters fixed. The analysis covers government inequality aversion, household preferences for school quality, sensitivity to housing prices, randomness in residential choice, and the contribution of student composition to school quality. Average access remains stable when inequality aversion, school-quality preference, and housing-price sensitivity are varied (MA: 0.5083--0.5137), while the equity measures respond more clearly. For example, increasing inequality aversion from 0 to 2 reduces GA from 0.0512 to 0.0261 without materially changing MA. This result confirms that the government objective can place greater weight on households with lower access without substantially reducing the overall access level.

The remaining parameters reveal socially meaningful trade-offs. Stronger housing-price sensitivity raises access inequality while leaving average access nearly unchanged. Greater randomness in residential choice reduces ID (from 0.1253 to 0.0637) but also lowers MA (from 0.5101 to 0.4857). Increasing the contribution of student composition similarly lowers GA but produces a limited reduction in MA. Thus, the main access result is stable under several behavioral assumptions, whereas the distributional outcomes appropriately depend on how households respond to school quality and housing costs. Appendix~\ref{app:experiment.sensitivity} reports all parameter settings and results.

We also evaluate six environments containing 12 or 24 communities and 4, 8, or 12 schools. Increasing the number of communities has the largest effect on average access (approximately 0.477 to 0.368), indicating that maintaining educational access becomes more difficult as households are distributed across a larger residential system. Increasing the number of schools has a smaller effect on MA but a clearer effect on inequality and sorting. With 12 communities, increasing the number of schools from 4 to 12 raises GA from 0.030 to 0.084 and Income Dissimilarity from 0.145 to 0.258.

These results show that expanding the number of schools does not by itself guarantee broader or more equal educational opportunity. In a residence-dependent system, a larger set of schools can create more differentiated access patterns and greater scope for household sorting. Educational planning should therefore consider the spatial distribution of households and schools together rather than treating school expansion as an independent solution. Complete results are provided in Appendix~\ref{app:scalability}.

\section{Disucssions and Conclusions} \label{sec:conclusions}

%In this work, we have introduced a computational framework that captures the complex socio-economic loop for educational investment allocation and links government decision, household sorting, school dynamics, housing market, and population turnover together to provide a holist understanding and provide insights for effective policy making. Important for policy experiments. provided a computational model to achieve this. We also provides RL algorithms for policy learning. and experiments have shown xxx phenomenon similar to social common sense, which demonstrate the effectiveness of the framework. 

In this work, we presented a dynamic multi-agent framework for allocating public-school investment under endogenous household responses. Its main contribution is to connect a tractable household-sorting equilibrium with sequential government decisions: investment changes school quality, which affects housing prices, residential sorting, enrollment, and future allocation needs. We further formulate the government problem as a POMDP and use RL to learn budget-feasible policies that balance aggregate access and equity. In simulations, the learned policy achieves the highest access level and the second-lowest access Gini among representative baselines, showing that explicitly accounting for the education--housing feedback can improve long-term allocation decisions. 
The PPO policy achieves a favorable balance between efficiency and equity by reducing disparities in school quality while preserving socioeconomic matching.
This promotes more equitable access to schools while increasing per-student investment. 
%This enables lower GA while improving MA and HC.
Scalability experiments further show that spatial expansion poses greater challenges to accessibility, whereas increasing school quantity alone may intensify sorting and reduce equity.
Overall, the proposed framework provides a computational foundation for evaluating educational investments whose effectiveness and distributional consequences evolve through household behavior over time.

%limitation: population structure change is not precisely modeled, exceed the scope and complex. limited education resource allocation only funding. only public school, no private, no transfer student considered. assumption that only school-entry group compete for housing, no transfer student considered.

Several extensions provide directions for future work. The current setting focuses on public-school funding with fixed attendance boundaries and residential decisions made at school entry. This allows us to isolate the central investment--sorting feedback without incorporating every education and housing mechanism into one model. Future work could introduce student transfers and richer demographic dynamics, private or charter-school options, and additional policy instruments such as teacher allocation and boundary adjustment. Calibrating these extensions with regional data would also support location-specific policy analysis.

\bibliography{aaai2027}

@article{chetty2016effects,
  title={The effects of exposure to better neighborhoods on children: New evidence from the moving to opportunity experiment},
  author={Chetty, Raj and Hendren, Nathaniel and Katz, Lawrence F},
  journal={American Economic Review},
  volume={106},
  number={4},
  pages={855--902},
  year={2016},
  publisher={American Economic Association 2014 Broadway, Suite 305, Nashville, TN 37203}
}

@article{biasi2023school,
  title={School finance equalization increases intergenerational mobility},
  author={Biasi, Barbara},
  journal={Journal of Labor Economics},
  volume={41},
  number={1},
  pages={1--38},
  year={2023},
  publisher={The University of Chicago Press Chicago, IL}
}

@article{owens2018income,
  title={Income segregation between school districts and inequality in students’ achievement},
  author={Owens, Ann},
  journal={Sociology of education},
  volume={91},
  number={1},
  pages={1--27},
  year={2018},
  publisher={Sage Publications Sage CA: Los Angeles, CA}
}

@article{black1999better,
  title={Do better schools matter? Parental valuation of elementary education},
  author={Black, Sandra E},
  journal={The quarterly journal of economics},
  volume={114},
  number={2},
  pages={577--599},
  year={1999},
  publisher={MIT Press}
}

@article{caetano2019neighborhood,
  title={Neighborhood sorting and the value of public school quality},
  author={Caetano, Gregorio},
  journal={Journal of Urban Economics},
  volume={114},
  pages={103193},
  year={2019},
  publisher={Elsevier}
}

@article{barseghyan2019peer,
  title={Peer preferences, school competition, and the effects of public school choice},
  author={Barseghyan, Levon and Clark, Damon and Coate, Stephen},
  journal={American Economic Journal: Economic Policy},
  volume={11},
  number={4},
  pages={124--158},
  year={2019},
  publisher={American Economic Association}
}

@article{abdulkadirouglu2003school,
  title={School choice: A mechanism design approach},
  author={Abdulkadiro{\u{g}}lu, Atila and S{\"o}nmez, Tayfun},
  journal={American economic review},
  volume={93},
  number={3},
  pages={729--747},
  year={2003},
  publisher={American Economic Association}
}

@article{owens2016income,
  title={Income segregation between schools and school districts},
  author={Owens, Ann and Reardon, Sean F and Jencks, Christopher},
  journal={American Educational Research Journal},
  volume={53},
  number={4},
  pages={1159--1197},
  year={2016},
  publisher={Sage Publications Sage CA: Los Angeles, CA}
}

@article{monarrez2023school,
  title={School attendance boundaries and the segregation of public schools in the United States},
  author={Monarrez, Tom{\'a}s E},
  journal={American Economic Journal: Applied Economics},
  volume={15},
  number={3},
  pages={210--237},
  year={2023},
  publisher={American Economic Association 2014 Broadway, Suite 305, Nashville, TN 37203-2425}
}

@article{bayer2007unified,
  author  = {Bayer, Patrick and Ferreira, Fernando and McMillan, Robert},
  title   = {A Unified Framework for Measuring Preferences for Schools and Neighborhoods},
  journal = {Journal of Political Economy},
  year    = {2007},
  volume  = {115},
  number  = {4},
  pages   = {588--638},
  doi     = {10.1086/522381}
}

@techreport{agostinelli2024spatial,
  author      = {Agostinelli, Francesco and Luflade, Margaux and Martellini, Paolo},
  title       = {On the Spatial Determinants of Educational Access},
  institution = {National Bureau of Economic Research},
  type        = {Working Paper},
  number      = {32246},
  year        = {2024},
  doi         = {10.3386/w32246}
}

@article{dignum2022mechanisms,
  author  = {Dignum, Eric and Athieniti, Efi and Boterman, Willem and Flache, Andreas and Lees, Michael},
  title   = {Mechanisms for Increased School Segregation Relative to Residential Segregation: A Model-Based Analysis},
  journal = {Computers, Environment and Urban Systems},
  year    = {2022},
  volume  = {93},
  pages   = {101772},
  doi     = {10.1016/j.compenvurbsys.2022.101772}
}

@article{lafortune2018school,
  author  = {Lafortune, Julien and Rothstein, Jesse and Schanzenbach, Diane Whitmore},
  title   = {School Finance Reform and the Distribution of Student Achievement},
  journal = {American Economic Journal: Applied Economics},
  year    = {2018},
  volume  = {10},
  number  = {2},
  pages   = {1--26},
  doi     = {10.1257/app.20160567}
}

@article{jackson2024impacts,
  author  = {Jackson, C. Kirabo and Mackevicius, Claire L.},
  title   = {What Impacts Can We Expect from School Spending Policy? Evidence from Evaluations in the United States},
  journal = {American Economic Journal: Applied Economics},
  year    = {2024},
  volume  = {16},
  number  = {1},
  pages   = {412--446},
  doi     = {10.1257/app.20220279}
}

@article{mayerle2022optimal,
  author  = {Mayerle, S{\'e}rgio F. and Rodrigues, Hidelbrando F. and de Figueiredo, Jo{\~a}o Neiva and De Genaro Chiroli, Daiane M.},
  title   = {Optimal Student/School/Class/Teacher/Classroom Matching to Support Efficient Public School System Resource Allocation},
  journal = {Socio-Economic Planning Sciences},
  year    = {2022},
  volume  = {83},
  pages   = {101341},
  doi     = {10.1016/j.seps.2022.101341}
}

@article{gillani2023redrawing,
  author  = {Gillani, Nabeel and Beeferman, Doug and Vega-Pourheydarian, Christine and Overney, Cassandra and Van Hentenryck, Pascal and Roy, Deb},
  title   = {Redrawing Attendance Boundaries to Promote Racial and Ethnic Diversity in Elementary Schools},
  journal = {Educational Researcher},
  year    = {2023},
  volume  = {52},
  number  = {6},
  pages   = {348--364},
  doi     = {10.3102/0013189X231170858}
}

@article{guan2025contextual,
  author  = {Guan, Hongzhao and Gillani, Nabeel and Simko, Tyler and Mangat, Jasmine and Van Hentenryck, Pascal},
  title   = {Contextual Stochastic Optimization for School Desegregation Policymaking},
  journal = {Proceedings of the AAAI Conference on Artificial Intelligence},
  year    = {2025},
  volume  = {39},
  number  = {27},
  pages   = {28024--28032},
  doi     = {10.1609/aaai.v39i27.35020}
}

@article{zhang2026deep,
  author  = {Zhang, Di and Mu, Senlin and Mango, Joseph and Li, Xiang},
  title   = {Deep Reinforcement Learning for Spatial Resource Allocation: A Case Study of School Districting},
  journal = {Environment and Planning B: Urban Analytics and City Science},
  year    = {2026},
  volume  = {53},
  number  = {2},
  pages   = {418--434},
  doi     = {10.1177/23998083241302187}
}

@article{abdulkadiroglu2020parents,
  author  = {Abdulkadiro{\u{g}}lu, Atila and Pathak, Parag A. and Schellenberg, Jonathan and Walters, Christopher R.},
  title   = {Do Parents Value School Effectiveness?},
  journal = {American Economic Review},
  year    = {2020},
  volume  = {110},
  number  = {5},
  pages   = {1502--1539},
  doi     = {10.1257/aer.20172040}
}

@article{turnbull2021meta,
  author  = {Turnbull, Geoffrey K. and Zheng, Minrong},
  title   = {A Meta-Analysis of School Quality Capitalization in {U.S.} House Prices},
  journal = {Real Estate Economics},
  year    = {2021},
  volume  = {49},
  pages   = {1120--1171},
  doi     = {10.1111/1540-6229.12300}
}

@article{owens2017racial,
  author  = {Owens, Ann},
  title   = {Racial Residential Segregation of School-Age Children and Adults: The Role of Schooling as a Segregating Force},
  journal = {RSF: The Russell Sage Foundation Journal of the Social Sciences},
  year    = {2017},
  volume  = {3},
  number  = {2},
  pages   = {63--80},
  doi     = {10.7758/RSF.2017.3.2.03}
}

@article{fernandez1996income,
  author  = {Fern{\'a}ndez, Raquel and Rogerson, Richard},
  title   = {Income Distribution, Communities, and the Quality of Public Education},
  journal = {The Quarterly Journal of Economics},
  year    = {1996},
  volume  = {111},
  number  = {1},
  pages   = {135--164},
  doi     = {10.2307/2946660}
}

@article{nechyba1999school,
  author  = {Nechyba, Thomas J.},
  title   = {School Finance Induced Migration and Stratification Patterns: The Impact of Private School Vouchers},
  journal = {Journal of Public Economic Theory},
  year    = {1999},
  volume  = {1},
  number  = {1},
  pages   = {5--50},
  doi     = {10.1111/1097-3923.00002}
}

@article{dignum2024data,
  author  = {Dignum, Eric and Boterman, Willem and Flache, Andreas and Lees, Mike},
  title   = {A Data-Driven Agent-Based Model of Primary School Segregation in Amsterdam},
  journal = {The Journal of Mathematical Sociology},
  year    = {2024},
  volume  = {48},
  number  = {3},
  pages   = {362--392},
  doi     = {10.1080/0022250X.2024.2340136}
}

@article{jackson2016effects,
  author  = {Jackson, C. Kirabo and Johnson, Rucker C. and Persico, Claudia},
  title   = {The Effects of School Spending on Educational and Economic Outcomes: Evidence from School Finance Reforms},
  journal = {The Quarterly Journal of Economics},
  year    = {2016},
  volume  = {131},
  number  = {1},
  pages   = {157--218},
  doi     = {10.1093/qje/qjv036}
}

@article{fernandez1998public,
  author  = {Fern{\'a}ndez, Raquel and Rogerson, Richard},
  title   = {Public Education and Income Distribution: A Dynamic Quantitative Evaluation of Education-Finance Reform},
  journal = {American Economic Review},
  year    = {1998},
  volume  = {88},
  number  = {4},
  pages   = {813--833}
}

@techreport{bayer2004equilibrium,
  author      = {Bayer, Patrick and McMillan, Robert and Rueben, Kim},
  title       = {An Equilibrium Model of Sorting in an Urban Housing Market},
  institution = {National Bureau of Economic Research},
  type        = {Working Paper},
  number      = {10865},
  year        = {2004},
  doi         = {10.3386/w10865}
}

@article{epple1999estimating,
  author  = {Epple, Dennis and Sieg, Holger},
  title   = {Estimating Equilibrium Models of Local Jurisdictions},
  journal = {Journal of Political Economy},
  year    = {1999},
  volume  = {107},
  number  = {4},
  pages   = {645--681},
  doi     = {10.1086/250074}
}

@article{schulman2017proximal,
  title={Proximal policy optimization algorithms},
  author={Schulman, John and Wolski, Filip and Dhariwal, Prafulla and Radford, Alec and Klimov, Oleg},
  journal={arXiv preprint arXiv:1707.06347},
  year={2017}
}

@inproceedings{mcfadden1972conditional,
  title={Conditional logit analysis of qualitative choice behavior},
  author={McFadden, Daniel},
  year={1972}
}

@misc{ceps2015,
  author       = {{National Survey Research Center}},
  title        = {{China Education Panel Survey (CEPS), 2013--2014 Baseline Survey}},
  year         = {2015},
  howpublished = {Chinese National Survey Data Archive},
  url          = {http://ceps.ruc.edu.cn/}
}

% Check whether the conference requires a reproducibility checklist to be included in the paper.
% If so, you can uncomment the following line and ajust the path to include it.
%\input{ReproducibilityChecklist.tex}

\newpage
\onecolumn
\appendix
\setcounter{secnumdepth}{2}

\section{Model Specifications} \label{app:model-spec}

\subsection{Household Parameters} \label{app:model-spec.household_param}
%Parental education increases the household's preference for school quality, while income affects its sensitivity to housing prices. 
%
In Equation~\eqref{eq:utility}, we set $\alpha_i = \alpha (1+\xi z_i )$, which is increasing in parental education and measures household $i$'s preference for school quality.
Price sensitivity is $\beta_i = \beta_0/y_i$, so an equal price increase places a greater utility burden on a lower-income household. 
Here, $\alpha$, $\xi$, and $\beta_0$ are tunable parameters. 
This is the main channel from school investment to socioeconomic sorting: improved quality can increase local demand and prices, thereby limiting access for households with fewer resources. Parameter $\eta$ weights fixed neighborhood amenities, and $\tau$ controls the dispersion of unobserved preferences. As $\tau\to0$, choices concentrate on the highest deterministic utility; larger $\tau$ gives flatter choice probabilities. All taste shocks are independent across households, alternatives, and periods. 
%Housing prices and moving costs apply only to movers because incumbents retain their units.
%Preschool in-migrants are myopic with respect to future school quality in their current location choice. They choose using~\eqref{eq:incoming-utility}, compete for current vacancies, and may reconsider their locations at school entry. 
Child ability $e_i$ is persistent for a household throughout the simulation. 

% \subsubsection{\red{Household-Type Construction. and CEPS data processing}}
% The finite type distribution $p_\theta$ is constructed from income, parental education, and child information in the China Education Panel Survey~\cite{ceps2015}.
% % TODO: Add the final CEPS preprocessing, discretization, and calibration procedure.

% \paragraph{Government parameters.}
% Discount factor $\gamma$ determines the weight on future access, $\epsilon$ controls inequality aversion, and $\lambda$ controls allocation stability. We set $\mathbf{b}_{-1}=\mathbf0$ at initialization. The base action contains only school investment; admission quotas, private-school choice, and residential transfers by already enrolled households are outside the model.

\subsection{Demographic Transitions} \label{app:model-spec.population}

The population transition from the end of period $t$ to the start of period $t+1$ proceeds in three steps:
\begin{enumerate}
    \item \emph{Aging.} Set $c_i^{t+1}=c_i^t+1$ for every resident child.
    \item \emph{Departures.} Each eligible non-school-age household exits according to the specified departure process. Enrolled households cannot exit during the $W$-period schooling window. Each departure releases one unit in the household's community.
    \item \emph{In-migration.} Draw incoming preschool-age household types from $p_\theta$. In period $t+1$, their locations are determined jointly with school-entry movers through the lower-level sorting equilibrium rather than assigned exogenously.
\end{enumerate}
At the beginning of period $t+1$, the transition identifies the school-entry group $\mathcal I^{t+1}$, the in-migrant group $\mathcal J^{t+1}$, and community vacancies $\mathbf V^{t+1}$. The government observes the resulting school and community summaries before selecting $\mathbf{b}^{t+1}$. In-migration is scaled so that aggregate expected fresh demand does not exceed vacancies and community capacity is preserved.

%\subsection{Special Cases on Zero School Enrollment}
\subsection{School Enrollment Specifications}
%The persistence parameter $\delta$ controls how quickly investment and peer composition replace inherited quality. Function $g_v$ captures diminishing returns to per-student investment, and $w_v$ and $w_s$ weight investment and peer composition. The exponent $\psi$ in~\eqref{eq:outcome} controls how persistent ability enters the secondary learning-outcome index.

%Every school has positive enrollment initially. If a school later becomes empty, its reference enrollment $\widetilde N_k^t$ and composition $\widetilde S_k^t$ are carried forward from its most recent nonempty cohort; otherwise, $\widetilde N_k^t=N_k^t$ and $\widetilde S_k^t=S_k^t$. Appendix~\ref{app:model-details} provides the formal convention. \red{app}

Every school has positive enrollment at initialization. However, free competitions allow households decide to not enroll certain schools, resulting in zero enrollment. In this case, we use the last non-zero enrollment number to help government decide the investment allocation. 
We define $\ell_k(t)=\max\{s \leq t: N_k^s > 0\}$. If $N^t_k = 0$ for some school $k$ and $t > 0$, we set the reference enrollment and composition as
\begin{equation}
\label{eq:empty-school}
    \widetilde N_k^t=N_k^{\ell_k(t)},
    \qquad
    \widetilde S_k^t=S_k^{\ell_k(t)}.
\end{equation}
Thus, current values are used for a nonempty school, while an empty school uses its most recent nonempty data. Total funding remains $g_k^t = b_k^t B^t$, per-student investment is $h_k^t = g_k^t / \widetilde N_k^t$.

\section{Sorting-Equilibrium Characterization} \label{app:equilibrium}

\subsection{Logit Demand Derivation} \label{app:equilibrium.logit}
For household $i$, the logit demand on a community is the probability distribution of applying the softmax function to the utilities of all communities. It originates from the discrete choice theory in \cite{mcfadden1972conditional}. We briefly introduce theory here.
% The logit demand uses softmax as a probability distribution to select communities. from \cite{mcfadden1972conditional}.

Suppose an agent chooses an object $j$ from a finite set $S$ based on the utility $U_j = V_j + \varepsilon_j$, where $V_j$ is a deterministic quantity and $\varepsilon_j$ is random. The utility is defined for every object $j \in S$, and we assume that $\varepsilon_j$ are i.i.d. for all $j$. 
%Suppose an agent chooses one from a set $S$. The agent evaluate every choice $j \in S$ with her utility function and chooses the hightest one. Let $U_j = V_j + \varepsilon_j$ be the utility function, where $\varepsilon_j$ is a random variable and $\varepsilon_j$ are i.i.d. for all $j$. The agent chooses with the highest utility. 
The probability of the agent choosing the object $i$ can be written as
\begin{equation}
    \operatorname{Pr}(U_i > U_j, \forall j) \quad \Rightarrow \quad 
    \operatorname{Pr}(V_i + \varepsilon_i > V_j + \varepsilon_j, \forall j) \quad \Rightarrow \quad 
    \operatorname{Pr}(\varepsilon_j < \varepsilon_i + V_i-V_j, \forall j):=\operatorname{Pr}(i \text{ chosen}).
\end{equation}
%which is equivalent to $\operatorname{Pr}(\varepsilon_j < \varepsilon_i + V_i-V_j, \forall j)$.
%
% Suppose the agent chooses the object $i$, meaning that $U_i$ is the largest among all $U_j$ for $j \in S, j\neq i$. Then we have
% \begin{equation}
%     \operatorname{Pr}(U_i > U_j, \forall j) \quad \Rightarrow \quad 
%     \operatorname{Pr}(V_i + \varepsilon_i > V_j + \varepsilon_j, \forall j) \quad \Rightarrow \quad 
%     \operatorname{Pr}(\varepsilon_j < \varepsilon_i + V_i-V_j, \forall j).
% \end{equation}
% Thus we need to compute this probability. 
%
When $\varepsilon_j$, $j \in S$, follows a Gumbel distribution $\operatorname{Gumbel}(0,\tau)$, its CDF and PDF are given by 
\begin{equation*}
    F(x)=e^{-e^{-x/\tau}}, \quad f(x) = \frac{1}{\tau} e^{-(x/\tau + e^{-x/\tau})}.
\end{equation*}
%is $F(x)=e^{-e^{-x/\tau}}$ and PDF as $f(x) = \frac{1}{\tau} e^{-(x/\tau + e^{-x/\tau})}$. 
Since all $\varepsilon_j$ are independent, when $i$ is chosen conditioned on the value of $\varepsilon_i = t$, we have
%every other random variables must satisfy $\varepsilon_j < t + V_i - V_j$. 
% \begin{equation*}
%     \varepsilon_j < t + V_i - V_j.
% \end{equation*}
%Since all $\varepsilon_j$ are independent, we have
\begin{equation}
\label{eq:prob_i1}
    \operatorname{Pr}(i \text{ chosen}) = \int_{-\infty}^\infty f(t) \prod_{j \neq i} F(t + V_i - V_j) dt
    = \int_{-\infty}^\infty \frac{1}{\tau} e^{-t/\tau} e^{-e^{t/\tau} \sum_{j} e^{(V_j-V_i)/\tau}} dt.
\end{equation}
% Therefore, taking the formula into the equation, we obtain
% \begin{equation}
%     \operatorname{Pr}(i \text{ chosen}) = \int_{-\infty}^\infty \frac{1}{\tau} e^{-t/\tau} e^{-e^{t/\tau} \sum_{j} e^{(V_j-V_i)/\tau}} dt
% \end{equation}
Using a change of variable $y = e^{-t}$, we can simplify \eqref{eq:prob_i1} and obtain
\begin{equation}
    \label{eq:prob_i}
    \operatorname{Pr}(i \text{ chosen}) = \int_0^\infty e^{-y \sum_{j} e^{(V_j-V_i)/\tau}} dy = \frac{1}{\sum_{j} e^{(V_j-V_i)/\tau}} = \frac{e^{V_i / \tau}}{\sum_j e^{V_j/\tau}}.
\end{equation}

Checking utilities \eqref{eq:utility}-\eqref{eq:incoming-utility} and using the result in \eqref{eq:prob_i}, we obtain the logit demand results in \eqref{eq:logit}-\eqref{eq:bg}.

\subsection{Convex-Program Characterization}
\label{app:equilibrium.convex}

We characterize the sorting-equilibrium price vector as the solution of a convex program. Let $\mathbf{P}=(P_m)_{m\in\mathcal M^t}$ denotes the prices of communities with vacancies. Prices outside $\mathcal M^t$ do not enter the sorting problem. Define the potential function
\begin{equation}
\label{eq:potential}
\begin{split}
    \Phi^t(\mathbf{P}) =
    \sum_{i\in\mathcal I^t}\frac{\tau}{\beta_i}
    \log\sum_{\ell\in\mathcal A_i^t}
    \exp\left(\bar u_{i\ell}^t(\mathbf{P})/\tau\right) 
    +
    \sum_{i\in\mathcal J^t}\frac{\tau}{\beta_i}
    \log\sum_{\ell\in\mathcal M^t}
    \exp\left(\bar v_{i\ell}^t(\mathbf{P})/\tau\right)
    +\sum_{m\in\mathcal M^t}V_m^tP_m.
\end{split}
\end{equation}
The first two terms are the log-sum-exp potentials associated with school-entry households and in-migrants, respectively. The last term accounts for the available housing supply.

\begin{proposition}[Convex-program characterization]
\label{prop:exist}
Under Assumption~\ref{ass:slack}, suppose that $\tau>0$ and $\beta_i>0$ for every active household. A feasible price vector $\mathbf{P}^{t\star}$ satisfies the market-clearing conditions in~\eqref{eq:clearing} if and only if
\begin{equation}
\label{eq:convex-program}
    \mathbf{P}^{t\star}
    \in
    \arg\min_{\mathbf{P}\geq P_{\min}\mathbf 1}
    \Phi^t(\mathbf{P}).
\end{equation}
A minimizer exists if
\begin{equation}
\label{eq:strict-residual-capacity}
    \sum_{m\in\mathcal M^t}V_m^t
    >
    |\mathcal J^t|.
\end{equation}
The minimizer is unique whenever $\Phi^t$ is strictly convex on the constraint set. In particular, a sufficient condition for strict convexity is that $\mathcal J^t$ is nonempty and at least one household in $\mathcal I^t$ has a feasible destination other than its current community.
\end{proposition}

\begin{proof}
We first establish convexity and derive the gradient of $\Phi^t$. From~\eqref{eq:utility} and~\eqref{eq:incoming-utility}, for $m\in\mathcal M^t$,
\begin{align}
    \frac{\partial \bar u_{i\ell}^t(\mathbf{P})}
    {\partial P_m}
    &=
    -\beta_i
    \mathbf 1(\ell=m)
    \mathbf 1(m\neq m_i^t),
    \label{eq:utility-price-derivative}\\
    \frac{\partial \bar v_{i\ell}^t(\mathbf{P})}
    {\partial P_m}
    &=
    -\beta_i\mathbf 1(\ell=m).
    \label{eq:incoming-price-derivative}
\end{align}
Thus, each deterministic utility is affine in $\mathbf{P}$. Because log-sum-exp is convex and $\tau/\beta_i>0$, the first two terms in~\eqref{eq:potential} are convex. The supply term is linear. Hence, $\Phi^t$ is convex.

Consider the contribution of household $i\in\mathcal I^t$ to the potential. Using~\eqref{eq:utility-price-derivative} and the logit probability in~\eqref{eq:logit}, we obtain
\begin{equation}
\label{eq:entry-potential-derivative}
\begin{split}
    &\frac{\partial}{\partial P_m}
    \left[
        \frac{\tau}{\beta_i}
        \log\sum_{\ell\in\mathcal A_i^t}
        \exp\left(\bar u_{i\ell}^t(\mathbf{P})/\tau\right)
    \right] \\
    =&
    \frac{\tau}{\beta_i}
    \frac{
        \exp\left(\bar u_{im}^t(\mathbf{P})/\tau\right)
    }{
        \sum_{\ell\in\mathcal A_i^t}
        \exp\left(\bar u_{i\ell}^t(\mathbf{P})/\tau\right)
    }
    \frac{-\beta_i}{\tau}
    \mathbf 1(m\neq m_i^t)
    =
    -\sigma_{im}^t(\mathbf{P};\mathbf r^t)
    \mathbf 1(m\neq m_i^t).
\end{split}
\end{equation}
Similarly, for an in-migrant $i\in\mathcal J^t$, equations~\eqref{eq:incoming-price-derivative} and~\eqref{eq:bg} give
\begin{align}
\label{eq:incoming-potential-derivative}
    \frac{\partial}{\partial P_m}
    \left[
        \frac{\tau}{\beta_i}
        \log\sum_{\ell\in\mathcal M^t}
        \exp\left(\bar v_{i\ell}^t(\mathbf{P})/\tau\right)
    \right]
    =
    -\mu_{im}^t(\mathbf{P}).
\end{align}
Combining~\eqref{eq:entry-potential-derivative} and~\eqref{eq:incoming-potential-derivative} with the derivative of the supply term yields
\begin{align}
\label{eq:potential-gradient}
    \frac{\partial\Phi^t(\mathbf{P})}{\partial P_m}
    =
    V_m^t
    -
    \sum_{\substack{i\in\mathcal I^t\\m\neq m_i^t}}
    \sigma_{im}^t(\mathbf{P};\mathbf r^t)
    -
    \sum_{i\in\mathcal J^t}
    \mu_{im}^t(\mathbf{P})
    =
    V_m^t-D_m^t(\mathbf{P};\mathbf r^t).
\end{align}
Thus, the demand function in~\eqref{eq:demand} arises directly from the gradient of the potential.

Next, we show the equivalence between the convex program and the market-clearing conditions. Let $\lambda_m \geq 0$ be the multiplier associated with the constraint $P_{\min}-P_m\leq0$. We obtain the Lagrangian
\begin{equation*}
    \mathcal{L}(\mathbf{P},\boldsymbol\lambda) =
    \Phi^t(\mathbf{P}) + \sum_{m\in\mathcal M^t} \lambda_m(P_{\min}-P_m).
\end{equation*}
The KKT conditions are therefore
\begin{align}
    P_m&\geq P_{\min},
    \label{eq:kkt-primal}\\
    \lambda_m&\geq0,
    \label{eq:kkt-dual}\\
    V_m^t-D_m^t(\mathbf{P};\mathbf r^t)-\lambda_m&=0,
    \label{eq:kkt-stationarity}\\
    \lambda_m(P_m-P_{\min})&=0.
    \label{eq:kkt-complementarity}
\end{align}
Eliminating $\lambda_m$ using~\eqref{eq:kkt-stationarity} yields
\begin{equation*}
    D_m^t(\mathbf{P};\mathbf r^t)\leq V_m^t,
    \quad
    P_m\geq P_{\min}, 
    \quad
    (P_m-P_{\min}) [V_m^t-D_m^t(\mathbf{P};\mathbf r^t)] = 0,
\end{equation*}
which are exactly the market-clearing conditions in~\eqref{eq:clearing}. Note that the constraint set $\mathcal{C} = \{\mathbf{P}: \mathbf{P} \geq P_{\min} \mathbf{1} \}$ is a polyhedron and the Slater's condition clearly holds. Therefore, the KKT conditions \eqref{eq:kkt-primal}-\eqref{eq:kkt-complementarity} are necessary and sufficient for optimality. Hence, a feasible price vector solves~\eqref{eq:convex-program} if and only if it is a sorting-equilibrium price vector.

Next, we show existence. Since the constraint set is closed but unbounded, convexity and continuity alone do not guarantee that the minimum is attained. Thus, we need to show that $\Phi^t$ is coercive on the constraint set because a continuous coercive function can attain its minimum on a nonempty closed subset of a finite-dimensional space. That is to show
\[
    \Phi^t(\mathbf P)\longrightarrow+\infty
    \quad\text{whenever}\quad
    \|\mathbf P\|_2\longrightarrow\infty,
    \qquad
    \mathbf P\geq P_{\min}\mathbf 1.
\]
For each $i\in\mathcal I^t$, staying in $m_i^t$ is feasible, and its utility does not depend on the prices of vacant units. Therefore, we have
\begin{equation}
\label{eq:entry-potential-lower-bound}
    \frac{\tau}{\beta_i}
    \log\sum_{\ell\in\mathcal A_i^t}
    \exp\left(\bar u_{i\ell}^t(\mathbf P)/\tau\right)
    \geq
    \frac{\bar u_{i m_i^t}^t}{\beta_i}.
\end{equation}
For each $i\in\mathcal J^t$, choose $k\in\arg\min_{\ell\in\mathcal M^t}P_\ell$. The log-sum-exp is at least as large as any one of its terms, so
\begin{equation}
\label{eq:incoming-potential-lower-bound}
\begin{split}
    \frac{\tau}{\beta_i}
    \log\sum_{\ell\in\mathcal M^t}
    \exp\left(\bar v_{i\ell}^t(\mathbf P)/\tau\right)
    \geq
    \frac{\eta Q_k}{\beta_i}-P_k
    \geq
    \min_{\ell\in\mathcal M^t}
    \frac{\eta Q_\ell}{\beta_i}
    -
    \min_{\ell\in\mathcal M^t}P_\ell.
\end{split}
\end{equation}
Consequently, there exists a finite constant $C^t$, independent of $\mathbf P$, such that
\begin{equation}
\label{eq:potential-lower-bound}
    \Phi^t(\mathbf P)
    \geq
    C^t
    +
    \sum_{m\in\mathcal M^t}V_m^tP_m
    -
    |\mathcal J^t|
    \min_{m\in\mathcal M^t}P_m.
\end{equation}
The lower bound in \eqref{eq:potential-lower-bound} diverges as prices become unbounded. To see this, let
\begin{equation*}
    P_-=\min_{m\in\mathcal M^t}P_m,
    \quad
    P_+=\max_{m\in\mathcal M^t}P_m,
    \quad 
    a=\sum_{m\in\mathcal M^t}V_m^t-|\mathcal J^t|>0,
    \qquad
    \underline{V} = \min_{m\in\mathcal M^t}V_m^t>0.
\end{equation*}
Then we have
\begin{align*}
    \sum_mV_m^tP_m-|\mathcal J^t|P_-
    &=
    aP_-+\sum_mV_m^t(P_m-P_-)\\
    &\geq
    aP_{\min}
    +\min\{a,\underline V\}(P_+-P_{\min}).
\end{align*}
Since the number of communities is finite and prices are bounded below by $P_{\min}$, $\|\mathbf P\|_2\to\infty$ implies $P_+\to\infty$. Hence, \eqref{eq:potential-lower-bound} implies $\Phi^t(\mathbf P) \to +\infty$ as $\|\mathbf P\|_2\to\infty$ within the constraint set. Thus, $\Phi^t$ is coercive. Since it is continuous and the constraint set is nonempty and closed, the convex program \eqref{eq:convex-program} has a minimizer.

Finally, we prove the stated sufficient condition for uniqueness. For $i\in\mathcal I^t$, define the vector of priced-move probabilities
\[
    \widetilde{\boldsymbol\sigma}_i^t(\mathbf{P})
    =
    \left(
        \sigma_{im}^t(\mathbf{P};\mathbf r^t)
        \mathbf 1(m\neq m_i^t)
    \right)_{m\in\mathcal M^t}.
\]
For $i\in\mathcal J^t$, let
\[
    \boldsymbol\mu_i^t(\mathbf{P})
    =
    \left(
        \mu_{im}^t(\mathbf{P})
    \right)_{m\in\mathcal M^t}.
\]
Differentiating~\eqref{eq:potential-gradient} again gives
\begin{equation}
\label{eq:potential-hessian}
\begin{split}
    \nabla^2 \Phi^t(\mathbf{P})
    =
    \sum_{i\in\mathcal I^t}
    \frac{\beta_i}{\tau}
    \left[
        \operatorname{Diag}
        \left(\widetilde{\boldsymbol\sigma}_i^t\right)
        -
        \widetilde{\boldsymbol\sigma}_i^t
        \left(\widetilde{\boldsymbol\sigma}_i^t\right)^\top
    \right]
    +
    \sum_{i\in\mathcal J^t}
    \frac{\beta_i}{\tau}
    \left[
        \operatorname{Diag}
        \left(\boldsymbol\mu_i^t\right)
        -
        \boldsymbol\mu_i^t
        \left(\boldsymbol\mu_i^t\right)^\top
    \right].
\end{split}
\end{equation}
All probabilities in the relevant choice sets are strictly positive because $\tau>0$ and utilities are finite. For any direction
$\mathbf d\in\mathbb R^{|\mathcal M^t|}$,
\begin{equation}
\label{eq:hessian-quadratic-form}
    \mathbf d^\top\nabla^2\Phi^t(\mathbf{P})\mathbf d
    =
    \sum_{i\in\mathcal I^t}
    \frac{\beta_i}{\tau}
    \left[
        \sum_{m\in\mathcal M^t}
        \widetilde{\sigma}_{im}^t d_m^2
        -
        \left(
            \sum_{m\in\mathcal M^t}
            \widetilde{\sigma}_{im}^t d_m
        \right)^2
    \right]
    +
    \sum_{i\in\mathcal J^t}
    \frac{\beta_i}{\tau}
    \left[
        \sum_{m\in\mathcal M^t}
        \mu_{im}^t d_m^2
        -
        \left(
            \sum_{m\in\mathcal M^t}
            \mu_{im}^t d_m
        \right)^2
    \right].
\end{equation}
%Each bracketed expression is a variance and is therefore nonnegative.
The bracketed expression for an in-migrant is the variance of $d_m$ under the probability distribution $\boldsymbol\mu_i^t$. For a school-entry household, it is also a variance after assigning the remaining probability $1-\sum_{m\in\mathcal M^t}\widetilde\sigma_{im}^t$ to the staying option, whose value in the price direction is zero. Therefore, every bracketed expression is nonnegative.

Suppose first that $\mathbf d$ is not constant across communities. Because every in-migrant assigns positive probability to every community in $\mathcal M^t$, the variance in the second sum is strictly positive for any $i\in\mathcal J^t$. Hence,
\[
    \mathbf d^\top\nabla^2\Phi^t(\mathbf{P})\mathbf d>0.
\]
Now suppose that $\mathbf d=c\mathbf 1$ for some $c\neq0$. Let $i\in\mathcal I^t$ be a household with at least one feasible move, and define its total probability of moving by
\[
    q_i^t(\mathbf{P})
    =
    \sum_{\substack{m\in\mathcal M^t\\m\neq m_i^t}}
    \sigma_{im}^t(\mathbf{P};\mathbf r^t).
\]
Because both staying and at least one moving option have positive probability,
\[
    0<q_i^t(\mathbf{P})<1.
\]
The contribution of this household to~\eqref{eq:hessian-quadratic-form} is
\[
    \frac{\beta_i}{\tau}
    c^2q_i^t(\mathbf{P})
    \left[1-q_i^t(\mathbf{P})\right]
    >0.
\]
Thus,
\[
    \mathbf d^\top\nabla^2\Phi^t(\mathbf{P})\mathbf d>0
    \qquad
    \text{for every }\mathbf d\neq\mathbf0.
\]
The Hessian is positive definite, so $\Phi^t$ is strictly convex. Therefore, its minimizer is unique.
\end{proof}

\subsection{Projected Gradient Method For Equilibrium Computation} \label{app:equilibrium.pgd}
We denote the constraint set in \eqref{eq:convex-program} as $\mathcal{C} = \{\mathbf{P}:\mathbf{P}\geq P_{\min}\mathbf 1\}$. Given its simplicity, we can use projected gradient descent method to compute the minimizer of \eqref{eq:convex-program} (i.e., the equilibrium price vector).
Projection onto this set can be computed componentwise:
\begin{equation*}
    \Pi_{\mathcal C}(\mathbf x) = \max\{P_{\min}\mathbf 1, \mathbf{x}\}.
\end{equation*}
Using~\eqref{eq:potential-gradient}, one projected-gradient step is
\begin{equation}
\label{eq:projected-price-update}
\begin{split}
    \mathbf{P}^{(s+1)} =
    \Pi_{\mathcal C}
    \left(
        \mathbf{P}^{(s)} -\zeta \nabla \Phi^t(\mathbf{P}^{(s)})
    \right)
    =
    \max \left\{
        P_{\min}\mathbf 1, \mathbf{P}^{(s)} + \zeta
        \left[
            \mathbf D^t(\mathbf{P}^{(s)};\mathbf r^t) - \mathbf V^t
        \right]
    \right\},
\end{split}
\end{equation}
where $\zeta>0$ is the step size. The update raises the price of a community when demand exceeds its vacancies. It lowers the price when vacancies exceed demand, subject to the price floor.

Define the projected-gradient residual as
\begin{equation}
\label{eq:projected-gradient-residual}
    \operatorname{res}_\zeta (\mathbf{P}) =
    \frac{1}{\zeta}
    \left[
        \mathbf{P} - \Pi_{\mathcal C}
        \left(
            \mathbf{P} - \zeta \nabla \Phi^t(\mathbf{P})
        \right)
    \right].
\end{equation}
The residual satisfies $\operatorname{res}_\zeta(\mathbf{P}) = \mathbf{0}$ if and only if $\mathbf{P}$ satisfies the KKT conditions and hence the market-clearing conditions. Thus, we propose Algorithm~\ref{alg:sorting} to compute the sorting equilibrium. The convergence result is established in Proposition~\ref{prop:projected-gradient-convergence} .

\begin{algorithm}[H]
\caption{Compute the lower-level sorting equilibrium}
\label{alg:sorting}
\begin{algorithmic}[1]
\REQUIRE School quality $\mathbf r^t$, vacancies $\mathbf V^t$, groups $\mathcal I^t$ and $\mathcal J^t$, feasible initial prices $\mathbf{P}^{(0)}$, price floor $P_{\min}$, fixed step size $\zeta$, and tolerance $\varepsilon_{\rm tol}$
\STATE Set $s\leftarrow0$
\REPEAT
\FOR{$m\in\mathcal M^t$}
\STATE Compute $D_m^t(\mathbf{P}^{(s)};\mathbf r^t)$ using
\eqref{eq:logit}--\eqref{eq:demand}
\ENDFOR
\STATE Update $\mathbf{P}^{(s+1)}$ using
\eqref{eq:projected-price-update}
\STATE Set
$\mathrm{res}^{(s)}
\leftarrow
\|\mathbf{P}^{(s+1)}-\mathbf{P}^{(s)}\|_\infty/\zeta$
\STATE Set $s\leftarrow s+1$
\UNTIL{$\mathrm{res}^{(s-1)}\leq\varepsilon_{\rm tol}$}
\STATE Set $\mathbf{P}^{t\star}\leftarrow\mathbf{P}^{(s)}$
\STATE Compute $\boldsymbol\sigma^{t\star}$ and
$\boldsymbol\mu^{t\star}$ using
\eqref{eq:logit}--\eqref{eq:bg}
\RETURN $\mathbf{P}^{t\star}$,
$\boldsymbol\sigma^{t\star}$,
$\boldsymbol\mu^{t\star}$
\end{algorithmic}
\end{algorithm}

\begin{proposition}[Convergence of projected gradient descent]
\label{prop:projected-gradient-convergence}
Suppose that the convex program~\eqref{eq:convex-program} has a
minimizer and that $\nabla\Phi^t$ is $L$-Lipschitz continuous on
$\mathcal C$ for some $L>0$. Starting from any
$\mathbf{P}^{(0)}\in\mathcal C$, let the fixed step size satisfy
\[
    0<\zeta\leq\frac{1}{L}.
\]
Then the sequence generated by~\eqref{eq:projected-price-update}
converges to a minimizer of~\eqref{eq:convex-program}. Moreover, for
any minimizer $\mathbf{P}^{t\star}$ and any $S\geq1$,
\begin{equation}
\label{eq:projected-gradient-rate}
    \Phi^t(\mathbf{P}^{(S)})
    -
    \Phi^t(\mathbf{P}^{t\star})
    \leq
    \frac{
        \|\mathbf{P}^{(0)}-\mathbf{P}^{t\star}\|_2^2
    }{
        2\zeta S
    }.
\end{equation}
If the minimizer is unique, then
\[
    \mathbf{P}^{(S)}
    \longrightarrow
    \mathbf{P}^{t\star}.
\]
\end{proposition}

\begin{proof}
We first give an explicit Lipschitz bound for the gradient. For any nonnegative vector $\mathbf p$ satisfying
\begin{equation*}
    q:=\sum_m p_m\leq1
\end{equation*}
and any vector $\mathbf d$, the Cauchy--Schwarz inequality gives
\begin{equation*}
    \left(\sum_m p_md_m\right)^2
    \leq
    \left(\sum_m p_m\right)
    \left(\sum_m p_md_m^2\right)
    =
    q\sum_m p_md_m^2
    \leq
    \sum_m p_md_m^2.
\end{equation*}
Therefore,
\begin{equation*}
    0
    \leq
    \mathbf d^\top
    \left[
        \operatorname{Diag}(\mathbf p)
        -
        \mathbf p\mathbf p^\top
    \right]
    \mathbf d
    \leq
    \sum_m p_m d_m^2
    \leq
    \|\mathbf d\|_2^2.
\end{equation*}
For every $i\in\mathcal I^t$, the vector $\widetilde{\boldsymbol\sigma}_i^t$ in \eqref{eq:potential-hessian} is a sub-probability vector, while for every $i\in\mathcal J^t$, the vector $\boldsymbol\mu_i^t$ is a probability vector. Hence, every matrix appearing in \eqref{eq:potential-hessian} is positive semidefinite and has spectral norm at most one. It follows that
\begin{equation*}
    \|\nabla^2\Phi^t(\mathbf P)\|_2
    \leq
    \frac{1}{\tau}
    \left(
        \sum_{i\in\mathcal I^t}\beta_i
        +
        \sum_{i\in\mathcal J^t}\beta_i
    \right)
    =:L_0.
\end{equation*}
Thus, $\nabla\Phi^t$ is globally Lipschitz continuous. In particular, when $L_0>0$, one may take
\begin{equation}
\label{eq:gradient-lipschitz-bound}
    L=L_0
    =
    \frac{1}{\tau}
    \left(
        \sum_{i\in\mathcal I^t}\beta_i
        +
        \sum_{i\in\mathcal J^t}\beta_i
    \right).
\end{equation}
If $L_0=0$, then $\nabla\Phi^t$ is constant, and any $L>0$ is a valid Lipschitz constant.

Let $\mathbf{P}^{t\star}$ be any minimizer. Since $\mathbf{P}^{(0)}\in\mathcal C$ and every update is a projection onto $\mathcal C$, all iterates belong to $\mathcal C$. The optimality condition for projection onto the closed convex set $\mathcal C$ gives, for every $\mathbf P\in\mathcal C$,
\begin{equation}
\label{eq:projection-optimality}
    \left\langle
        \mathbf{P}^{(s+1)}
        -
        \left[
            \mathbf{P}^{(s)}
            -
            \zeta\nabla\Phi^t(\mathbf{P}^{(s)})
        \right],
        \mathbf P-\mathbf{P}^{(s+1)}
    \right\rangle
    \geq0.
\end{equation}
Setting $\mathbf P=\mathbf{P}^{t\star}$ and using the identity
\begin{equation*}
    2\langle \mathbf a-\mathbf b,\mathbf c-\mathbf a\rangle
    =
    \|\mathbf b-\mathbf c\|_2^2
    -
    \|\mathbf a-\mathbf c\|_2^2
    -
    \|\mathbf a-\mathbf b\|_2^2
\end{equation*}
gives
\begin{equation}
\label{eq:projection-bound}
\begin{split}
    \left\langle
        \nabla\Phi^t(\mathbf{P}^{(s)}),
        \mathbf{P}^{(s+1)}-\mathbf{P}^{t\star}
    \right\rangle 
    \leq
    \frac{1}{2\zeta}
    \left[
        \|\mathbf{P}^{(s)}-\mathbf{P}^{t\star}\|_2^2
        -
        \|\mathbf{P}^{(s+1)}-\mathbf{P}^{t\star}\|_2^2
        -
        \|\mathbf{P}^{(s+1)}-\mathbf{P}^{(s)}\|_2^2
    \right].
\end{split}
\end{equation}
Also by convexity, we have
\begin{equation*}
    \Phi^t(\mathbf{P}^{(s)})
    -
    \Phi^t(\mathbf{P}^{t\star})
    \leq
    \left\langle
        \nabla\Phi^t(\mathbf{P}^{(s)}),
        \mathbf{P}^{(s)}-\mathbf{P}^{t\star}
    \right\rangle.
\end{equation*}
Combining this inequality with the descent lemma yields
\begin{equation}
\label{eq:descent-bound}
\begin{split}
    \Phi^t(\mathbf{P}^{(s+1)})
    -
    \Phi^t(\mathbf{P}^{t\star})
    \leq
    \left\langle
        \nabla\Phi^t(\mathbf{P}^{(s)}),
        \mathbf{P}^{(s+1)}-\mathbf{P}^{t\star}
    \right\rangle
    +
    \frac{L}{2}
    \|\mathbf{P}^{(s+1)}-\mathbf{P}^{(s)}\|_2^2.
\end{split}
\end{equation}
Combining~\eqref{eq:projection-bound} and \eqref{eq:descent-bound} gives
\begin{equation}
\label{eq:one-step-convergence}
\begin{split}
    \Phi^t(\mathbf{P}^{(s+1)})
    -
    \Phi^t(\mathbf{P}^{t\star}) 
    \leq
    \frac{1}{2\zeta}
    \left[
        \|\mathbf{P}^{(s)}-\mathbf{P}^{t\star}\|_2^2
        -
        \|\mathbf{P}^{(s+1)}-\mathbf{P}^{t\star}\|_2^2
    \right] 
    -
    \left(
        \frac{1}{2\zeta}-\frac{L}{2}
    \right)
    \|\mathbf{P}^{(s+1)}-\mathbf{P}^{(s)}\|_2^2.
\end{split}
\end{equation}
Since $\zeta\leq1/L$, the final term is nonpositive and can be removed from the upper bound. Summing over $s=0,\ldots,S-1$ yields
\[
    \sum_{s=0}^{S-1}
    \left[
        \Phi^t(\mathbf{P}^{(s+1)})
        -
        \Phi^t(\mathbf{P}^{t\star})
    \right]
    \leq
    \frac{
        \|\mathbf{P}^{(0)}-\mathbf{P}^{t\star}\|_2^2
    }{
        2\zeta
    }.
\]

Next, we show that the objective values are nonincreasing. Setting $\mathbf P=\mathbf{P}^{(s)}$ in \eqref{eq:projection-optimality} gives
\begin{equation*}
    \left\langle
        \nabla\Phi^t(\mathbf{P}^{(s)}),
        \mathbf{P}^{(s+1)}-\mathbf{P}^{(s)}
    \right\rangle
    \leq
    -\frac{1}{\zeta}
    \|\mathbf{P}^{(s+1)}-\mathbf{P}^{(s)}\|_2^2,
\end{equation*}
which implies
\begin{equation*}
    \Phi^t(\mathbf{P}^{(s+1)})
    \leq
    \Phi^t(\mathbf{P}^{(s)})
    -
    \left(
        \frac{1}{\zeta}-\frac{L}{2}
    \right)
    \|\mathbf{P}^{(s+1)}-\mathbf{P}^{(s)}\|_2^2.
\end{equation*}
Since $\zeta\leq1/L$, the coefficient is nonnegative. Thus,
$\{\Phi^t(\mathbf{P}^{(s)})\}$ is nonincreasing. Consequently,
\begin{equation*}
\begin{split}
    S\left[
        \Phi^t(\mathbf{P}^{(S)})
        -
        \Phi^t(\mathbf{P}^{t\star})
    \right]
    \leq
    \sum_{s=0}^{S-1}
    \left[
        \Phi^t(\mathbf{P}^{(s+1)})
        -
        \Phi^t(\mathbf{P}^{t\star})
    \right] 
    \leq
    \frac{
        \|\mathbf{P}^{(0)}-\mathbf{P}^{t\star}\|_2^2
    }{
        2\zeta
    },
\end{split}
\end{equation*}
which proves~\eqref{eq:projected-gradient-rate}.

Finally, we show convergence of the iterates. Rearranging~\eqref{eq:one-step-convergence} gives
\begin{equation*}
    \|\mathbf P^{(s+1)}-\mathbf P^{t\star}\|_2^2
    \leq
    \|\mathbf P^{(s)}-\mathbf P^{t\star}\|_2^2 
    -
    2\zeta
    \left[
        \Phi^t(\mathbf P^{(s+1)})
        -
        \Phi^t(\mathbf P^{t\star})
    \right] 
    -
    (1-\zeta L)
    \|\mathbf P^{(s+1)}-\mathbf P^{(s)}\|_2^2.
\end{equation*}
Since $\zeta\leq1/L$ and $\Phi^t(\mathbf{P}) \geq \Phi^t(\mathbf{P}^*) \ \forall \mathbf{P} \in \mathcal{C}$, both subtracted terms are nonnegative and we have
\begin{equation*}
    \|\mathbf P^{(s+1)}-\mathbf P^{t\star}\|_2
    \leq
    \|\mathbf P^{(s)}-\mathbf P^{t\star}\|_2.
\end{equation*}
Thus, the sequence $\{P^{(s)}\}$ is Fej\'er monotone with respect to the set of minimizers and is therefore bounded and has a convergent subsequence that belongs to $\mathcal{C}$ since $\mathcal{C}$ is closed.
Moreover, by \eqref{eq:projected-gradient-rate} we have
\[
    \Phi^t(\mathbf P^{(S)})
    \longrightarrow
    \min_{\mathbf P\in\mathcal C}\Phi^t(\mathbf P).
\]
Then, the limit of any convergent subsequence converges to the minimum of $\Phi^t$ since $\Phi^t$ is continuous.
%Since the sequence is bounded, it has a convergent subsequence. The limit of any convergent subsequence belongs to $\mathcal C$ and attains the minimum of $\Phi^t$ because $\mathcal C$ is closed and $\Phi^t$ is continuous. Thus, every subsequential limit is a minimizer. 
The standard convergence theorem for Fej\'er-monotone sequences in finite-dimensional spaces then implies that the full sequence $\mathbf P^{(s)}$ converges to a minimizer. If the minimizer is unique, the limit is $\mathbf P^{t\star}$.

%Since $\mathcal C$ is closed and $\Phi^t$ is continuous, every cluster point of the sequence is a minimizer. The standard convergence theorem for Fej\'er-monotone sequences in finite-dimensional spaces therefore implies that $\mathbf P^{(s)}$ converges to a minimizer. If the minimizer is unique, the limit is $\mathbf P^{t\star}$.
\end{proof}

\subsection{From Equilibrium Probabilities to Assignments}
\label{app:equilibrium.assignment}

The sorting equilibrium $(\mathbf P^{t\star},\boldsymbol\sigma^{t\star}, \boldsymbol\mu^{t\star})$ specifies household choice probabilities and clears the vacancy market in expectation. 
We apply a capacity-feasible realization procedure to determine a discrete residential assignment for every household in each period. The enrollment and school composition are computed based on the assignment result. Throughout the procedure, the prices $\mathbf P^{t\star}$ remain fixed.

\paragraph{Step 1: Initial Realization.} Each household in group $\mathcal{I}$ and $\mathcal{J}$ independently sample a tentative destination from its equilibrium probabilities:
\begin{equation*}
\label{eq:initial-mover-draw}
    a_i^{t,0}
    \sim
    \operatorname{Categorical}
    \left(
        \{\sigma_{im}^{t\star}\}_{m\in\mathcal A_i^t}
    \right),
    \ i\in\mathcal I^t, 
    \qquad
    a_j^{t,0}
    \sim
    \operatorname{Categorical}
    \left(
        \{\mu_{jm}^{t\star}\}_{m\in\mathcal M^t}
    \right),
    \ j\in\mathcal J^t.
\end{equation*}
They form an assignment profile $\mathbf{a}^t$.

\paragraph{Step 2: Local Overflow Rationing}
For an assignment profile $\mathbf a^t$, define the fresh-arrival set of community $m$ as
\begin{equation*}
\label{eq:fresh-arrival-set}
    \mathcal{F}_m^t(\mathbf a^t) = \{i\in\mathcal I^t: a_i^t=m,\ m\neq m_i^t\} 
    \cup 
    \{j\in\mathcal J^t: a_j^t=m\}.
\end{equation*}
For each community $m$, let $x_m^t = |\mathcal F_m^t(\mathbf a^{t})|$ denote its realized fresh demand, and define its overflow by $\xi_m^t = (x_m^t-V_m^t)^+$. If $\xi_m^t=0$, all tentative fresh arrivals to $m$ are accepted. If $\xi_m^t > 0$, the procedure selects a subset
\begin{equation*}
    \mathcal R_m^t \subseteq \mathcal F_m^t(\mathbf a^{t}), \quad 
    |\mathcal R_m^t|=\xi_m^t
\end{equation*}
uniformly among all subsets of this size. The selected households lose their tentative assignments and enter a pending pool. We denote the sets of displaced households as 
\begin{equation*}
    \widetilde{\mathcal{I}}^{t} = \left( \bigcup_m \mathcal{R}_m^t \right) \cap \mathcal{I}^t, 
    \quad 
    \widetilde{\mathcal{J}}^{t} = \left( \bigcup_m \mathcal{R}_m^t \right) \cap \mathcal{J}^t.
\end{equation*}
After placing the households in $\mathcal{F}_m^t(\mathbf a^t) \backslash \mathcal{R}_m^t$ for all $m$, the residual vacancies in community becomes 
\begin{equation*}
    \widetilde{V}_m^t = V_m^t - | \mathcal{F}_m^t(\mathbf a^t) \backslash \mathcal{R}_m^t |.
\end{equation*}

\paragraph{Step 3: Reassigning Displaced Households.}
The procedure processes households in $\widetilde{I}^t$ and $\widetilde{J}^t$ in uniformly random order. For a pending household $i \in \widetilde{I}^t$, define the currently available community set  $\widetilde{\mathcal{A}}_i^t = \{m_i^t\} \cup \{m \in \mathcal{M}^t: \widetilde{V}_m^t > 0 \}$. The household draws a new destination using its equilibrium probabilities conditioned on current
availability:
\begin{equation}
\label{eq:conditional-mover-draw}
    \Pr(a_i^t=m) =
    \frac{\sigma_{im}^{t\star}}
    {\sum_{\ell \in \widetilde{\mathcal{A}}_i^t}
        \sigma_{i\ell}^{t\star}},
    \quad
    m \in \widetilde{\mathcal A}_i^t.
\end{equation}
If $a_i^t \neq m_i^t$, the procedure decreases
$\widetilde V_{a_i^t}^t$ by one. If $a_i^t=m_i^t$, no vacancy is consumed. Thus, every displaced mover either reaches another community with available capacity or remains in its origin.
Similarly, for a pending household $j \in \widetilde{J}^t$, define the currently available community set $\widetilde{\mathcal{M}}_j^t = \{ m \in \mathcal{M}^t: \widetilde{V}_m^t > 0\}$. The household draws a new destination using its equilibrium probabilities conditioned on current availability:
\begin{equation}
\label{eq:conditional-migrant-draw}
    \Pr(a_j^t=m)
    =
    \frac{\mu_{jm}^{t\star}}
    {\sum_{\ell\in\widetilde{\mathcal M}_j^t}
        \mu_{j\ell}^{t\star}},
    \qquad
    m\in\widetilde{\mathcal M}_j^t.
\end{equation}
After assigning $j$, the procedure decreases
$\widetilde V_{a_j^t}^t$ by one. 

Availability is recomputed before assigning the next household. The procedure satisfies capacity-feasible and discrete placement for every realization. 
%, which is summarized in Algorithm~\ref{}.
The resulting assignment result update household locations and further determines the enrolled sets $\mathcal L_k^t$, enrollment $N_k^t$, and peer composition $S_k^t$.

\section{Government Decision Process and Algorithms}\label{app:gov-decision}

\subsection{Partially Observed Decision Process}
Let
\[
    \mathcal{G}=\langle \mathcal{X}, \Delta_K, p, R, \mathcal{Q}, \mathcal{O}, \gamma \rangle
\]
denote the government decision process. The full state $x^t\in\mathcal{X}$ contains household types, ages and locations; school quality; housing occupancy and vacancies; the active groups; and the previous allocation required to evaluate the smoothing penalty. Given action $\mathbf{b}^t$, the transition kernel $p(x^{t+1} | x^t, \mathbf{b}^t)$ includes sorting, assignment, school production, and demographic turnover.

The government observes $q^t=\mathcal{O}(x^t)$ as defined in~\eqref{eq:observation}. In the implementation, this observation contains school quality, community occupancy relative to capacity, community mean income and ability, and entry-group features. It does not expose the full household microstate. The policy is therefore observation-based, with Dirichlet concentration parameters $\boldsymbol\kappa_\phi(q^t) > \mathbf{0}$. The reward is
\begin{equation}
\label{eq:period-reward}
    R^t=\mathcal W_\epsilon(\boldsymbol\rho^{t,+})
    -\lambda\|\mathbf{b}^t-\mathbf{b}^{t-1}\|_2^2.
\end{equation}
The expectation in~\eqref{eq:gov-objective} is over demographic transitions, household taste shocks and realized assignments, and stochastic policy actions.

\subsection{Training Procedure}

We summarize the complete training procedure in Algorithm~\ref{alg:system}.

\begin{algorithm}[H]
\caption{Dynamic system and government-policy training}
\label{alg:system}
\begin{algorithmic}[1]
\REQUIRE Initial population $\mathcal H^0$, quality $\mathbf r^0$, horizon
$T$, budgets $\{B^t\}$, and parameters $\gamma,\epsilon,\lambda$
\STATE Initialize policy $\pi_\phi$, value function $V_\omega$, and
$\mathbf{b}^{-1}=\mathbf{0}$
\FOR{each training episode}
\STATE Reset the population and school qualities
\FOR{$t=0,\ldots,T-1$}
\STATE Construct $q^t=\mathcal{O}(x^t)$ from school and community summaries
\STATE Sample $\mathbf{b}^t\sim\pi_\phi(\cdot\mid q^t)$ and set
$g_k^t = b_k^t B^t$
\STATE Call Algorithm~\ref{alg:sorting} using inherited quality $\mathbf r^t$
\STATE Compute enrollment $N_k^t$ and composition $S_k^t$
\STATE Apply~\eqref{eq:empty-school} and set
$h_k^t=g_k^t/\widetilde N_k^t$
\STATE Update quality to $\mathbf r^{t+1}$ using~\eqref{eq:quality}
\STATE Compute $\boldsymbol\rho^{t,+}$, $\mathbf o^{t,+}$, and reward~\eqref{eq:period-reward}
\STATE Process turnover to form $x^{t+1}$
\STATE Store $(q^t,\mathbf{b}^t,R^t,q^{t+1})$
\ENDFOR
\STATE Update $\pi_\phi,V_\omega$ with PPO using discounted episode returns
\ENDFOR
\RETURN Trained policy $\pi_\phi$
\end{algorithmic}
\end{algorithm}

\section{Experiment Setting } \label{app:experiment}

\subsection{Data Source} \label{app:experiment.data}

We use the China Education Panel Survey (CEPS) to discipline the joint distribution of household characteristics in the simulation
\cite{ceps2015}.  The processed calibration sample contains 16,880 observations.
For each observation, the calibration artifact records household income, bthe child's ability, parental education level.  Thus, the survey is not used as a static set of transitions for reinforcement learning.  Instead, it determines the population distribution from which households are generated, while the state--action--next-state transitions used for policy learning are produced online by the dynamic simulator.

The empirical joint probability array has dimensions
% \begin{equation}
%   \begin{aligned}
%   &37\;\text{income levels}\times5\;\text{ability levels}
%    \times9\;\text{parental-education levels}.
%   \end{aligned}
% \end{equation}
\begin{center}
    37 income levels $\times$ 5 ability levels $\times$ 9 parental-education levels.
\end{center}
Let $Y$, $E$ and $L$denote the corresponding discrete random variables.  The normalized empirical frequency distribution is denoted by
$\widehat p(y,e,z)$ and satisfies
\begin{equation}
  \sum_{y,e,z,j,a}\widehat p(y,e,z)=1.
\end{equation}
Occupation and assets are retained in the calibration artifact but do not enter the current household state.  We therefore integrate out these two dimensions:
\begin{equation}
  p_{yez}
  =\Pr(Y=y,E=e,Z=z)
  =\sum_{j\in\mathcal J}\sum_{a\in\mathcal A}
    \widehat p(y,e,z,j,a).
  \label{eq:ceps-marginal}
\end{equation}
The resulting model type space contains $37 \times 5 \times 9 = 1665$ joint income--ability--education types.  Drawing types jointly from \eqref{eq:ceps-marginal}, rather than drawing the three characteristics from independent marginals, preserves their empirical dependence and, in particular, the income--ability sorting present in the survey.

Table~\ref{tab:ceps-variables} summarizes the variables entering the model.
Income has 37 observed support points ranging from 1 to 20 in the units of the processed survey variable.  The five academic-performance categories are mapped to the normalized ability scores
\begin{equation}
  e\in\{0.0430,\;0.1828,\;0.4528,\;0.7801,\;0.9670\}.
\end{equation}
Parental education consists, in ascending order, of no schooling, primary
school, junior high school, technical secondary/vocational school,
vocational high school, senior high school, junior college, university, and
postgraduate education or above.  These nine ordered categories are mapped
linearly onto
\begin{equation}
  z \in \{0,0.125,0.25,\ldots,1\}.
\end{equation}

\begin{table}[H]
  \centering
  \begin{tabular}{@{}llp{0.27\textwidth}p{0.38\textwidth}@{}}
    \toprule
    Variable & Symbol & Support & Role in the model \\
    \midrule
    Household income
      & $y_i$ & 37 levels, $[1,20]$
      & Price sensitivity $\beta_i=\beta_0/y_i$ \\
    Child ability
      & $e_i$ & 5 normalized levels in $[0,1]$
      & Peer composition and learning output \\
    Parental education
      & $z_i$ & 9 ordered levels in $[0,1]$
      & School-quality preference $\alpha_i=\alpha(1+\xi z_i)$ \\
    \bottomrule
  \end{tabular}
  \caption{CEPS variables used to construct the simulated household population.}
  \label{tab:ceps-variables}
\end{table}

\begin{comment}
\red{Parental education introduces observed heterogeneity in the value attached to
school quality.  In the baseline specification,
\begin{equation}
  \alpha_i=\alpha\bigl(1+\xi \ell_i\bigr),
  \qquad \alpha=1,\quad \xi=1,
  \label{eq:education-preference-data}
\end{equation}
so that the highest parental-education group places twice as much weight on
school quality as the lowest group, holding income and other characteristics
fixed.  Ability does not enter residential utility directly; it affects peer
composition and the educational production function.  This distinction
prevents the residential sorting mechanism from mechanically assigning a
direct taste for school quality to high-ability children.} 

\red{NEED DISCUSSION:  remove to Appendix A.1?}
\end{comment}

\subsection{Metric Definitions} \label{app:experiment.metrics}
This section provides the definitions of all evaluation metrics used in Section~8. These metrics quantify educational efficiency, equity, and socioeconomic sorting in the simulated education resource allocation environment. All metrics are computed based on the realized educational outcomes after household decisions and school quality evolution.

\subsubsection{Mean  Access(MA): }

The mean educational access measures the average quality received by students:

\begin{equation}
\bar{\rho}^t
=
\frac{1}{N^t}\sum_i\rho_i^t
=
\frac{\sum_{k=1}^{K}N_k^t r_k^t}{N^t}.
\end{equation}

Higher values indicate better average educational access. This metric is
equivalent to $W_0(\boldsymbol{\rho}^t)$.

\subsubsection{Gini Access(GA): }

Educational inequality is measured using the Gini coefficient of accessed
quality:

\begin{equation}
G_{\rho}^t
=
\frac{
\sum_{k=1}^{K}\sum_{\ell=1}^{K}
N_k^tN_\ell^t|r_k^t-r_\ell^t|
}
{2(N^t)^2\bar{\rho}^t}.
\end{equation}

The value satisfies $G_{\rho}^t\in[0,1]$. A value of zero indicates that all
students access the same educational quality, while larger values indicate
greater inequality. When $\bar{\rho}^t=0$, the implementation returns zero.

\subsubsection{Income Access Gap(IA): }

To measure socioeconomic differences in educational access, students are ranked
according to their income $y_i$. The lowest and highest income quartiles are
denoted as $\mathcal I_{Q_1}^t$ and $\mathcal I_{Q_4}^t$, respectively.

The mean accessed quality of each quartile group is:

\begin{equation}
\bar{\rho}_{Q_q}^t
=
\frac{4}{N^t}
\sum_{i\in\mathcal I_{Q_q}^t}\rho_i^t,
\qquad q\in\{1,4\}.
\end{equation}

The income-based access gap is defined as:

\begin{equation}
\Delta_{\rho,y}^t
=
\bar{\rho}_{Q_4}^t-\bar{\rho}_{Q_1}^t .
\end{equation}

A positive value indicates that high-income students access higher-quality
schools, while a value close to zero indicates similar educational access
between income groups.

The implementation handles discrete income types by splitting weights when a
quartile boundary falls inside an income category. The metric is undefined only
when no students are enrolled.

\subsubsection{Dissimilarity Income: }

Income segregation across schools is measured using the Duncan dissimilarity
index. Let $L_k^t$ and $H_k^t$ denote the numbers of low- and high-income
students enrolled in school $k$, respectively. The index is defined as:

\begin{equation}
D_y^t
=
\frac{1}{2}
\sum_{k=1}^{K}
\left|
\frac{L_k^t}{L^t}
-
\frac{H_k^t}{H^t}
\right|.
\end{equation}

The index satisfies $D_y^t\in[0,1]$. A value of zero indicates identical
distributions of income groups across schools, while larger values indicate
stronger income segregation.

\subsubsection{Income--Quality Correlation(IQ): }

To quantify the relationship between socioeconomic status and educational
quality, we compute the population-weighted Pearson correlation between
community income and assigned school quality.

For community $m$, let $n_m^t$ denote the realized population,
$\bar y_m^t$ denote the average income, and $r_{d(m)}^t$ denote the quality of
the assigned school. The population weight is defined as:

\begin{equation}
\omega_m^t
=
\frac{n_m^t}
{\sum_j n_j^t}.
\end{equation}

The weighted correlation coefficient is:

\begin{equation}
\operatorname{Corr}_y^t
=
\frac{
\sum_m\omega_m^t
(\bar y_m^t-\mu_y^t)
(r_{d(m)}^t-\mu_r^t)
}
{
\sqrt{\sum_m\omega_m^t(\bar y_m^t-\mu_y^t)^2}
\sqrt{\sum_m\omega_m^t(r_{d(m)}^t-\mu_r^t)^2}
}.
\end{equation}

A positive value indicates that higher-income communities tend to access
higher-quality schools, while a value close to zero indicates a weaker
relationship between income and educational quality.

\subsubsection{Human Capital(HC): }

Following Eq.~(11) of the paper, individual educational output is defined as:

\begin{equation}
o_i^t=\rho_i^t e_i^\psi .
\end{equation}

The aggregate human capital is calculated as:

\begin{equation}
H_t
=
\sum_{i:c_i\text{ is of school age}}
\rho_i^t e_i^\psi .
\end{equation}

Equivalently, when aggregated by ability type and school:

\begin{equation}
H_t
=
\sum_e\sum_{k=1}^{K}
N_{ek}^t r_k^t e^\psi .
\end{equation}

A larger value indicates higher aggregate educational output. Since this is a
total measure, it depends jointly on enrollment size, school quality, and the
matching between student ability and educational resources.

\subsection{Sensitivity Analysis Experiments}
\label{app:experiment.sensitivity}

To examine the robustness of the proposed framework, we conduct a series of one-factor-at-a-time sensitivity analyses. In each experiment, one parameter is varied while keeping all other parameters fixed at their baseline values. The selected parameters cover different aspects of the proposed model, including government preference, household residential behavior, and school quality evolution. Specifically, we investigate the impact of the inequality-aversion parameter $\epsilon$, the school-quality preference parameter $\alpha$, the housing price sensitivity parameter $\beta$, the residential choice randomness parameter $\tau$, and the peer-effect weight $w_s$. The sensitivity results demonstrate how these parameters influence the trade-off between educational efficiency, equity, and socioeconomic sorting.

\subsubsection{Sensitive Analysis.}

\paragraph{$\epsilon$ Analysis.}
The parameter $\epsilon$ controls the degree of inequality aversion in the
Atkinson welfare function. Table~\ref{tab:epsilon} reports
the sensitivity of the equilibrium outcomes with respect to $\epsilon$.

As $\epsilon$ increases, educational inequality decreases significantly. GA
decreases from 0.0512 at $\epsilon=0$ to 0.0261 at $\epsilon=2$, while IA moves
closer to zero, indicating improved equity across income groups. Meanwhile, IQ
decreases from -0.5183 to values close to zero, suggesting that educational
opportunities become less dependent on household socioeconomic status.

However, excessively large values of $\epsilon$ provide limited additional
equity benefits. When $\epsilon$ increases from 2 to 4, GA slightly increases
to 0.0308 and IA becomes positive, indicating potential over-correction. In
contrast, MA remains stable around 0.51 across different settings, while HC
exhibits only minor fluctuations. These results suggest that moderate
inequality aversion ($\epsilon\approx1$--$2$) achieves a favorable balance
between educational equity and overall performance.
\begin{table}[H]
\centering
%\small
\setlength{\tabcolsep}{3pt}
\begin{tabular}{lccccc}
\toprule
Metric & 0 & 0.5 & 1 & 2 & 4 \\
\midrule
MA         & 0.5088 & 0.5109 & 0.5096 & 0.5101 & 0.5092 \\
GA           & 0.0512 & 0.0364 & 0.0320 & 0.0261 & 0.0308 \\
IA     & -0.0181 & -0.0096 & 0.0007 & -0.0018 & 0.0032 \\
ID & 0.1216 & 0.1229 & 0.1366 & 0.1253 & 0.1400 \\
IQ  & -0.5183 & -0.3775 & 0.0024 & -0.0606 & 0.0976 \\
HC         & 94.4947 & 96.1598 & 94.7289 & 94.6490 & 95.4126 \\
\bottomrule
\end{tabular}
\caption{Sensitivity analysis with respect to $\epsilon$.}
\label{tab:epsilon}
\end{table}

\paragraph{$\alpha$ Analysis.}

The parameter $\alpha$ controls the sensitivity of households to school quality
differences in the residential choice process. Table~\ref{tab:alpha} reports the sensitivity
of equilibrium outcomes with respect to different values of $\alpha$.

The results show that MA remains relatively stable across different values of
$\alpha$ (0.5083--0.5137), indicating that school-quality preference has
limited impact on overall accessibility. However, it affects distributional
outcomes. GA decreases from 0.0381 at $\alpha=0.25$ to 0.0213 at
$\alpha=0.5$, but slightly increases for larger values, suggesting that
moderate quality preference improves equity, while excessive quality-driven
sorting may increase disparities.

IQ exhibits a non-monotonic pattern, decreasing from 0.4640 to negative values
at intermediate $\alpha$ and increasing to 0.3021 at $\alpha=4$. Meanwhile, HC
remains relatively stable (94.2901--96.2011), suggesting that $\alpha$ mainly
affects the distribution of educational opportunities rather than aggregate
educational production.

Overall, the sensitivity analysis shows that moderate school-quality
preference achieves a better balance between educational matching and equity,
whereas excessive quality-driven sorting may strengthen socioeconomic
differences in school access.
\begin{table}[H]
\centering
%\small
\setlength{\tabcolsep}{3pt}
\begin{tabular}{lccccc}
\toprule
Metric & 0.25 & 0.5 & 1 & 2 & 4 \\
\midrule
MA           & 0.5083 & 0.5102 & 0.5101 & 0.5137 & 0.5109 \\
GA         & 0.0381 & 0.0213 & 0.0261 & 0.0313 & 0.0260 \\
IA    & 0.0132 & -0.0006 & -0.0018 & -0.0037 & 0.0069 \\
ID   & 0.1279 & 0.1244 & 0.1253 & 0.1359 & 0.1452 \\
IQ  & 0.4640 & 0.0015 & -0.0606 & -0.1654 & 0.3021 \\
HC         & 94.9945 & 96.2011 & 94.6490 & 94.2901 & 95.2680 \\
\bottomrule
\end{tabular}
\caption{Sensitivity analysis with respect to $\alpha$.}
\label{tab:alpha}
\end{table}

\paragraph{ $\beta$ Analysis.}

The parameter $\beta$ captures the sensitivity of households to housing prices
in the residential sorting process. A larger $\beta$ indicates that housing
costs play a more important role in household location decisions, potentially
strengthening the interaction between housing markets and educational access.
Table~\ref{tab:beta} reports the sensitivity of equilibrium outcomes with
respect to different values of $\beta$.

When $\beta\leq1$, the equilibrium outcomes remain nearly unchanged. MA, GA,
IA, and HC are identical across $\beta=0.25$, $0.5$, and $1$, indicating that
moderate variations in housing price sensitivity have limited effects on
educational outcomes.

As $\beta$ increases beyond one, housing price sensitivity begins to affect
educational distribution. GA increases from 0.0261 at $\beta=1$ to 0.0377 at
$\beta=4$, suggesting that stronger housing price constraints amplify
inequality in educational access. Meanwhile, IA remains close to zero,
indicating limited changes in income-based access differences.

The impact of $\beta$ on efficiency-related outcomes is limited. MA decreases
only slightly from 0.5101 to 0.5093, and HC changes from 94.6490 to 94.5725.
These results suggest that housing price sensitivity mainly affects the
distribution of educational resources rather than aggregate performance.

Overall, the sensitivity analysis indicates that moderate housing price
sensitivity has limited influence on the equilibrium outcomes, whereas
excessive dependence on housing costs increases educational access inequality.
This highlights the role of housing markets as a potential channel through
which socioeconomic differences influence educational opportunities.
\begin{table}[H]
\centering
%\small
\setlength{\tabcolsep}{3pt}
\begin{tabular}{lccccc}
\toprule
Metric & 0.25 & 0.5 & 1 & 2 & 4 \\
\midrule
MA           & 0.5101 & 0.5101 & 0.5101 & 0.5096 & 0.5093 \\
GA          & 0.0261 & 0.0261 & 0.0261 & 0.0320 & 0.0377 \\
IA     & -0.0018 & -0.0018 & -0.0018 & 0.0016 & -0.0004 \\
ID  & 0.1253 & 0.1253 & 0.1253 & 0.1345 & 0.1170 \\
IQ  & -0.0606 & -0.0606 & -0.0606 & -0.0007 & -0.1015 \\
HC        & 94.6490 & 94.6490 & 94.6490 & 94.5882 & 94.5725 \\
\bottomrule
\end{tabular}
\caption{Sensitivity analysis with respect to $\beta$.}
\label{tab:beta}
\end{table}

\paragraph{ $\tau$ Analysis.}

The parameter $\tau$ controls the degree of randomness in household school-choice behavior. 
deterministic utility differences and introduces more stochasticity into the
residential sorting process. Table~\ref{tab:tau} reports the sensitivity of
equilibrium outcomes with respect to different values of $\tau$.

As $\tau$ increases, MA decreases gradually from 0.5101 to 0.4857, indicating
that stronger choice randomness reduces average accessibility. In contrast, HC
increases from 94.6490 to 104.4688, suggesting that less deterministic sorting
improves the utilization of educational resources.

The impact of $\tau$ on socioeconomic sorting is also significant. DI decreases
from 0.1253 to 0.0637, while IQ moves toward zero, indicating weaker
associations between income and school quality. However, GA does not decrease
monotonically, reaching its lowest value at $\tau=1$ and increasing slightly for
larger values.

Overall, the sensitivity analysis reveals an efficiency--equity trade-off
associated with $\tau$. Higher choice randomness reduces average educational
access but improves aggregate human capital and weakens socioeconomic sorting.
Moderate randomness ($\tau\approx1$--$2$) achieves a better balance between
educational efficiency and distributional outcomes.
\begin{table}[H]
\centering
%\small
\setlength{\tabcolsep}{3pt}
\begin{tabular}{lccccc}
\toprule
Metric & 0.25 & 0.5 & 1 & 2 & 4 \\
\midrule
MA           & 0.5101 & 0.5018 & 0.4959 & 0.4896 & 0.4857 \\
GA     & 0.0261 & 0.0267 & 0.0121 & 0.0432 & 0.0410 \\
IA    & -0.0018 & -0.0006 & -0.0004 & 0.0009 & 0.0002 \\
ID  & 0.1253 & 0.0842 & 0.0738 & 0.0624 & 0.0637 \\
IQ  & -0.0606 & -0.0958 & -0.0750 & -0.0504 & 0.0348 \\
HC        & 94.6490 & 97.6483 & 100.8233 & 101.5455 & 104.4688 \\
\bottomrule
\end{tabular}
\caption{Sensitivity analysis with respect to $\tau$.}
\label{tab:tau}
\end{table}

\paragraph{$(w_s,w_v)$ Analysis.}
The parameters $(w_s,w_v)$ determine the relative importance of student composition and government investment in the evolution of school quality. Specifically, $w_s$ captures the contribution of peer effects, while $w_v$ represents the effectiveness of financial investment in improving school quality. Table~\ref{tab:weight} reports the sensitivity of equilibrium outcomes under different combinations of these two weights.

The results show that increasing $w_s$ generally improves equity. As $w_s$
increases from 0 to 1, GA decreases from 0.0726 to 0.0212, indicating that
stronger peer effects reduce inequality in school-quality access.

In contrast, increasing $w_v$ mainly improves efficiency-related outcomes.
When the weight of government investment increases, HC rises from 94.1940 under
$(w_s,w_v)=(1,0)$ to 95.1909 under $(w_s,w_v)=(0,1)$. This occurs because a
larger $w_v$ strengthens the impact of resource investment on school quality.
However, prioritizing investment effectiveness may weaken equity improvements,
as reflected by higher GA.

Overall, the weighting parameters reveal a trade-off between efficiency and equity. Larger $w_v$ favors HC accumulation, whereas larger $w_s$ promotes more
equitable access through lower GA.

\begin{table}[H]
\centering
%\small
\setlength{\tabcolsep}{3pt}
\begin{tabular}{lccccc}
\toprule
Metric & (0,1) & (0.25,0.75) & (0.5,0.5) & (0.75,0.25) & (1,0) \\
\midrule
MA           & 0.5098 & 0.5108 & 0.5063 & 0.5016 & 0.4999 \\
GA           & 0.0726 & 0.0389 & 0.0273 & 0.0294 & 0.0212 \\
IA    & 0.0081 & -0.0104 & -0.0022 & 0.0075 & 0.0021 \\
ID   & 0.1352 & 0.1088 & 0.1302 & 0.1284 & 0.1143 \\
IQ & 0.1295 & -0.3919 & -0.1130 & 0.3483 & 0.0742 \\
HC    & 95.1909 & 93.1327 & 93.8852 & 94.2137 & 94.1940 \\
\bottomrule
\end{tabular}
\caption{Sensitivity analysis with respect to the weighting coefficients $(w_s, w_v)$.}
\label{tab:weight}
\end{table}

\subsection{Scalability Experiments} \label{app:scalability}
To evaluate the scalability of the proposed framework under different spatial
and educational configurations, we conduct experiments with varying numbers of
communities and schools. Specifically, we consider six configurations by
changing the number of communities from 12 to 24 and the number of schools from
4 to 12. Table~\ref{tab:config_results} reports the resulting performance under
different environment scales.

The results demonstrate that the number of communities has a more significant
impact on MA than the number of schools. When the number of communities
increases from 12 to 24, MA decreases substantially from approximately 0.477 to
0.368, while changing the number of schools within the same community scale has
only a limited effect (e.g., 0.477, 0.477, and 0.471 under 4, 8, and 12 schools
with 12 communities). This indicates that spatial expansion and household
dispersion introduce greater challenges for maintaining educational
accessibility than increasing the number of educational providers.

In contrast, increasing the number of schools mainly affects the distributional
properties of educational opportunities. As the number of schools increases,
GA tends to increase. For example, under 12 communities, GA rises from 0.030
with 4 schools to 0.084 with 12 schools. Similarly, DI increases from 0.145 to
0.258, indicating stronger socioeconomic segregation across schools. These
results suggest that simply increasing the number of schools does not
necessarily improve educational equity. A larger number of schools may provide
more differentiated educational environments, but can also intensify household
sorting and amplify disparities in school access.

The scalability experiments further reveal that the relationship between school
quantity and educational outcomes is not monotonic. While increasing the number
of schools may improve educational diversity, excessive fragmentation of
educational resources can strengthen sorting mechanisms and increase
inequality. Therefore, expanding school quantity alone is insufficient to
guarantee improved educational outcomes, and effective resource allocation
requires jointly considering spatial structure, household distribution, and
school configuration.

\begin{table}[H]
\centering
%\scriptsize
\setlength{\tabcolsep}{3pt}
\begin{tabular}{lcccccc}
\toprule
 & (12,4) & (12,8) & (12,12) & (24,4) & (24,8) & (24,12) \\
\midrule
MA
& 0.477 & 0.477 & 0.471 & 0.368 & 0.367 & 0.367 \\
GA
& 0.030 & 0.032 & 0.084 & 0.015 & 0.037 & 0.047 \\
IA
& 0.007 & -0.007 & -0.026 & 0.002 & -0.0003 & 0.007 \\
DI
& 0.145 & 0.195 & 0.258 & 0.085 & 0.160 & 0.226 \\
IQ
& 0.289 & -0.288 & -0.452 & 0.206 & -0.030 & 0.228 \\ 
HC
& 106.366 & 105.801 & 105.613 & 162.425 & 163.465 & 161.777 \\
\bottomrule
\end{tabular}
\caption{Performance under different community and school configurations.}
\label{tab:config_results}
\end{table}

\subsection{Comparative Experiments} \label{app:experiment.comparison}

In the comparative experiments, we evaluate the proposed reinforcement learning (RL)-based resource allocation policy against several representative baseline strategies, including \emph{equal-split budgeting}, \emph{proportional-to-enrollment allocation}, and \emph{compensatory funding}. The allocation rule for each baseline is defined as follows:

\begin{itemize}
    \item \textbf{Equal-split.} The total education budget is equally distributed
    across all schools,
   \begin{equation}
    b_k^t=\frac{1}{K}, \quad k=1,\dots,K.
    \end{equation}

    \item \textbf{Proportional-to-enrollment.} The budget is allocated according
    to school enrollment,
    \begin{equation}
    b_k^t=\frac{N_k^t}{\sum_{j=1}^{K}N_j^t},
    \end{equation}
    corresponding to equal per-student funding and representing a common
    allocation rule in public education systems.

    \item \textbf{Compensatory funding.} Schools with lower educational quality
    receive a larger budget share,
    \begin{equation}
    b_k^t=
    \frac{\exp(-\lambda r_k^t)}
    {\sum_{j=1}^{K}\exp(-\lambda r_j^t)},
    \end{equation}
    where $\lambda$ controls the compensation intensity and $r_k^t$ denotes
    the current quality of school $k$. This heuristic represents an
    equity-oriented funding strategy that prioritizes disadvantaged schools.
\end{itemize}
\subsection{Parameter Setting}
\label{app:parameter-setting}

\begin{table}[H]
  \centering
  %\small
  \begin{tabular}{@{}llcl@{}}
    \toprule
    Component & Parameter & Baseline & Interpretation \\
    \midrule
    Geography
      & \(M,K\) & \(12,4\) & Communities and schools \\
      & \(C_m\) & \(100\) & Housing capacity per community \\
      & \(Q_m\) & \(\mathcal U[0,1]\) & Initial community amenity \\
      & \(r_k^0\) & \(\mathcal U[0.1,0.9]\) & Initial school quality \\
    \addlinespace
    Household choice
      & \(\alpha\) & \(1.0\) & Baseline preference for school quality \\
      & \(\xi\) & \(1.0\) & Parental-education preference gradient \\
      & \(\beta_0\) & \(1.0\) & Price-sensitivity scale \\
      & \(D,\lambda\) & \(1.0,1.0\) & Relocation cost and its weight \\
      & \(\eta\) & \(1.0\) & Amenity weight \\
      & \(\tau\) & \(0.25\) & Logit temperature \\
      & \(P_{\min}\) & \(0\) & Housing-price floor \\
    \addlinespace
    School dynamics
      & \(\delta\) & \(0.4\) & Quality adjustment rate \\
      & \(w_v,w_s\) & \(0.8,0.2\) & Investment and peer weights \\
      & \(g(v)\) & \(\sqrt{v}\) & Concave investment technology \\
      & \(\psi\) & \(1.0\) & Ability exponent in learning output \\
    \addlinespace
    Government
      & \(B\) & \(100\) & Annual education budget \\
      & \(\epsilon\) & \(2.0\) & Atkinson inequality aversion \\
      & \(\kappa\) & \(0\) & Policy-stability penalty \\
      & \(\gamma\) & \(0.95\) & Intertemporal discount factor \\
    \addlinespace
    Population
      & \(\pi\) & \(0.2\) & Annual departure probability \\
      & \(a_{\mathrm{entry}},W\) & \(6,6\) & Entry age and schooling duration \\
      & \(\chi_{\mathrm{in}}\) & \(1.0\) & Inflow-to-departure ratio \\
      & \(\omega_0\) & \(0.95\) & Initial occupancy rate \\
    \addlinespace
    Simulation
      & \(T,B_{\mathrm{burn}}\) & \(50,10\) & Horizon and burn-in periods \\
      & \(S\) & \(100\) & Number of evaluation seeds \\
    \bottomrule
  \end{tabular}
  \caption{Baseline model and simulation parameters.}
  \label{tab:baseline-parameters}
\end{table}

% The policy is trained with proximal policy optimization (PPO) for
% \(200{,}000\) environment steps.  The policy and value networks use the
% default multilayer-perceptron architecture of Stable-Baselines3.
% Table~\ref{tab:ppo-parameters} lists the optimization hyperparameters.  The
% discount factor is inherited from the government objective rather than
% specified independently in the PPO configuration.  At the beginning of each
% training episode, the simulator draws a reproducible new population and
% market realization, while the configured geography, amenities, and initial
% school-quality support remain fixed.  This exposes the policy to multiple
% initial realizations and reduces dependence on a single training trajectory.

\begin{table}[H]
  \centering
  %\small
  \begin{tabular}{@{}lcl@{}}
    \toprule
    Hyperparameter & Value & Description \\
    \midrule
    Learning rate & \(3\times10^{-3}\) & Adam step size \\
    Rollout length & \(100\) & Environment steps per update \\
    Batch size & \(20\) & Minibatch size \\
    Training epochs & \(10\) & Passes over each rollout \\
    GAE coefficient & \(0.95\) & Bias--variance trade-off \\
    PPO clip range & \(0.2\) & Policy-update clipping threshold \\
    Entropy coefficient & \(0\) & Entropy regularization weight \\
    Discount factor & \(0.95\) & Inherited from the government objective \\
    Training steps & \(200{,}000\) & Total environment interactions \\
    \bottomrule
  \end{tabular}
  \caption{PPO training hyperparameters.}
  \label{tab:ppo-parameters}
\end{table}

% TODO: Add the final network architecture, PPO hyperparameters, training budget,
% and realized-assignment implementation details.
% \appendix
% - detailed parameter definitions in Agent models, including household utility model, government model, school model
% - details of demographic change module.
% - explanation and proof of proposition 1
% - emphasize price is the sufficient statistics in sorting equilibrium
% - alg to compute sorting equilibrium
% - alg of government rl training

\end{document}